\documentclass[aps,pra,twocolumn,showpacs,superscriptaddress,groupedaddress]{revtex4} 

\usepackage{graphicx}  
\usepackage{dcolumn}   
\usepackage{bm}        
\usepackage{amssymb}   
\usepackage{amsmath}
\usepackage{amsthm}

\newtheorem{lemma}{Lemma}

\newtheorem{dfn}{Definition}
\newtheorem{problem}{Problem}

\newtheorem{prop}{Proposition}

\newcommand{\calC}{{\cal C }}

\newcommand{\calL}{{\cal L }}
\newcommand{\calM}{{\cal M }}
\newcommand{\calE}{{\cal E }}

\newcommand{\calG}{{\cal G }}
\newcommand{\calF}{{\cal F }}

\newcommand{\calV}{{\cal V }}

\newcommand{\CC}{\mathbb{C}}
\newcommand{\RR}{\mathbb{R}}

\newcommand{\la}{\langle}
\newcommand{\ra}{\rangle}

\begin{document}

\title{Simulation of rare events in quantum error correction}

\author{Sergey \surname{Bravyi}}
\affiliation{IBM T.J. Watson Research Center, Yorktown Heights NY 10598}

\author{Alexander \surname{Vargo}}
\affiliation{IBM T.J. Watson Research Center, Yorktown Heights NY 10598}

\date{\today}

\begin{abstract}
We consider the problem of calculating  the logical error probability
for a stabilizer quantum code subject to  random Pauli errors.
To access the regime of large code distances where logical errors are extremely unlikely
we adopt the splitting method widely used in Monte Carlo simulations of
rare events and Bennett's acceptance ratio method for estimating the free energy difference
between two canonical ensembles. 
To illustrate the power of these methods in the context of error correction,
we calculate the logical error probability $P_L$
for the 2D surface code on a square lattice with a pair of holes for all code distances $d\le 20$ and all error
rates $p$ below the fault-tolerance threshold. Our numerical results confirm the expected
exponential decay $P_L\sim \exp{[-\alpha(p)d]}$ and provide a simple fitting formula 
for the decay rate $\alpha(p)$. Both noiseless and noisy syndrome readout circuits are considered. 
\end{abstract}

\maketitle

\section{Introduction}
\label{sec:intro}

Quantum error correction holds the promise of extending the coherence time of
quantum devices by utilizing redundant encoding of information.  Like its
classical counterpart, quantum error correction enables reliable storage of
encoded quantum states in the presence of noise by monitoring parity check
violations and applying suitable recovery operations.  Furthermore, many
quantum codes support a limited set of logical operations that can be applied
to encoded states without exposing them to noise.  Extensive theoretical work
rigorously confirmed the feasibility of large-scale fault-tolerant quantum
computing for a wide range of noise
models~\cite{Shor96,Knill04,AGP06,aharonov2006fault,ng2009fault}.

Several families of quantum codes have been proposed as candidates for
scalable fault-tolerant architectures, including concatenated
codes~\cite{aharonov1997fault,kitaev1997quantum,knill1998resilient},
surface codes~\cite{Kitaev03,Dennis01,BK:surface,Fowler08,RH:cluster2D},
surface codes with twists~\cite{BombinTwists}, color
codes~\cite{BMD:topo,LAR2011} and Turaev-Viro codes~\cite{Koenig10}.  Each
of these families contains an infinite sequence of codes labeled by a code
distance $d$.  The number of physical qubits and elementary operations
required to implement a single logical gate using a distance-$d$ code typically
grows polynomially as one increases $d$, whereas the probability of a logical
error $P_L$ decreases exponentially, that is, $P_L\le d^\beta \cdot
\exp{(-\alpha d)}$.  Here $\alpha$ and $\beta$ are constant coefficients
depending on the chosen family of codes, the noise model, and the decoding
algorithm. Furthermore,  the threshold theorem asserts that $\alpha>0$ for a
sufficiently small noise strength~\cite{AGP06,Dennis01,RHG07,fowler2012proof}.
A problem essential for estimating the overhead associated with error
correction is finding the minimum code distance that achieves the desired level
of noise suppression.  This requires a precise knowledge of the decay rate 
\[
\alpha=-\lim_{d\to \infty}\, \frac1d \log{(P_L(d))},
\]
since the exponential term gives the dominant contribution to $P_L$
for large code distances.  
The decay rate $\alpha$ is also a natural figure of merit for comparing
the performance of different decoding algorithms.

The present paper describes a new algorithm for computing the decay rate
$\alpha$ in the special case of the surface code family. We report numerical
results for two commonly studied noise models corresponding to noiseless and
noisy syndrome readout circuits.  Prior to our work, several methods have been
developed for computing the decay rate $\alpha$ of the surface codes, most
notably Monte Carlo simulation~\cite{RHG07} and fault path
counting~\cite{Dennis01,fowler2013analytic}.  Monte Carlo method attempts
to estimate $P_L$ by generating many random error configurations and computing
the fraction of trials that resulted in a logical error.  This method, however,
is not a viable option in the regime of small physical error rates and/or
large code distances, where logical errors are extremely unlikely. 

Fault path counting method adapted to surface codes by Dennis et
al~\cite{Dennis01} provides an upper bound on $P_L$ in the form of a weighted
sum over self-avoiding walks on a suitably defined lattice. This can be
translated to a lower bound on the decay rate $\alpha$, although the bound is
not expected to be tight.  A different version of the method, proposed by
Fowler~\cite{fowler2013analytic}, enables exact computation of $P_L$ in the
asymptotic regime of small error rates $p$.  In this regime the dominant
contribution to $P_L$ comes from minimum-weight uncorrectable errors that span
$\lceil d/2\rceil$ physical locations~\footnote{Any error of weight smaller
	than $d/2$ must be correctable by definition of the code distance.}. 
Accordingly, if the limit $p\to 0$ is taken for a fixed code distance $d$, one
can use an asymptotic formula $P_L=A_d\, p^{d/2}$, where $A_d$ is a constant
coefficient.  For relatively small values of $d$ one can compute $A_d$ by
summing the probabilities of all minimum-weight uncorrectable errors, see
Ref.~\cite{fowler2013analytic}.  This method, however, is not well-suited for
computing the decay rate $\alpha$ which requires taking the limit
$d\to \infty$ for a fixed error rate $p$.

Here we propose a new algorithm for estimating the logical error probability
$P_L$ and the decay rate $\alpha$.  It enables us, for the first time, to
access the regime of large code distances and moderately small error rates,
which we expect to be particularly important in the context of fault-tolerance.
The key ingredients of our algorithm are the splitting method and Bennett's
acceptance ratio method~\cite{bennett1976}.  The splitting method is a
standard tool in Monte Carlo simulation of rare events, see for instance
Ref.~\cite{rubino2009rare}. To compute the quantity $P_L(p)$ for a given
physical error rate $p$ we choose a sufficiently dense monotonic sequence of
error rates $p_1,\ldots,p_t$, where $p_t= p$ and $p_1$ is chosen such that
$P_L(p_1)$ can be computed efficiently by either the Monte Carlo simulation (in
which case $p_1$ must be sufficiently large) or by the fault path counting
method (in which case $p_1$ must be sufficiently small).  We then employ the
acceptance ratio method~\cite{bennett1976} to estimate the
quantity $R_j=P_L(p_{j+1})/P_L(p_j)$ for each $j=1,\ldots,t-1$.  This is accomplished
through a Metropolis-type subroutine for sampling uncorrectable errors from
a specific probability distribution. The Metropolis subroutine is the most time
consuming part of our algorithm and we propose several tricks for improving its
efficiency.  Finally, we find $P_L(p)$ from the obvious identity
$P_L(p)=P_L(p_t)=R_1 R_2 \cdots R_{t-1}P_L(p_1)$.  Once $P_L(p)$ has been
computed for several choices of the code distances $d$, the exponential fitting yields
the decay rate $\alpha(p)$.  Depending on whether the sequence $p_1,\ldots,p_t$
is increasing or decreasing we shall use a term upward or downward splitting.
To the best of our knowledge, the idea of using the splitting method in the
context of error correction was originally proposed by Wang, Harrington, and
Preskill~\cite{WHP2002}, see Section~IV(D) of Ref.~\cite{WHP2002}.  The use of
the acceptance ratio method in this context appears to be new.

\section{Summary of results}
\label{sec:results}

We begin by highlighting main features of the surface code and motivating its
choice as a testing ground for the proposed algorithm.  The surface code
achieves fault-tolerance by repeatedly measuring syndromes of 4-qubit parity
check operators on a 2D grid of physical qubits~\cite{Dennis01}.  Logical
qubits are introduced by creating special defect areas in the lattice in which
the pattern of syndrome measurements is altered.  By choosing a suitable
syndrome readout schedule one can change the location and shape of defects in a
way that simulates braiding, splitting, and fusion of topological
charges~\footnote{Note that a defect can carry a coherent
	superposition of different topological charges, as opposed to
	point-like anyons that must have a definite topological charge.}. 
By analogy with topological quantum computation~\cite{Kitaev03}, this enables
fault-tolerant implementation of some logical gates such as the
CNOT~\cite{BMD:codedef, RH:cluster2D,Fowler08}.  The code distance $d$ is
determined by the length of the shortest loop encircling a defect or the
shortest path connecting some pair of defects.  The code admits efficient
decoding by Edmonds's minimum weight matching
algorithm~\cite{Dennis01,Fowler12}, renormalization group
methods~\cite{Poulin09,DP10a1,BH11a}, 
or by the Markov chain Monte Carlo algorithm~\cite{Hutter13},
and features an error threshold close to
$1\%$ for the standard depolarizing noise model~\cite{Wang11}.

\begin{figure}[h]
\includegraphics[scale=0.5]{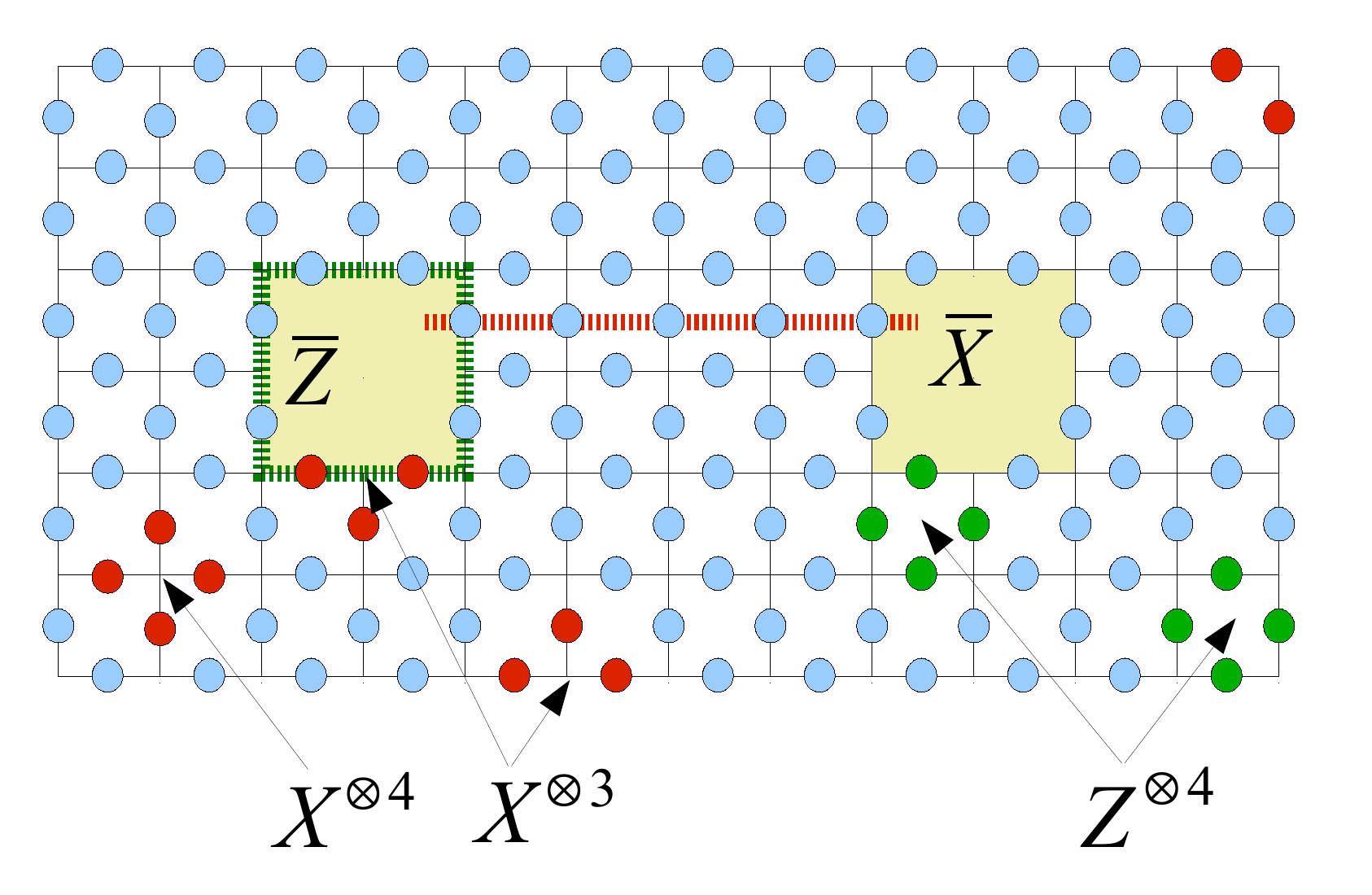}
\caption{\label{fig:flat}  (Color online) Surface code with  a pair of smooth defects
representing one logical qubit. 
Physical qubits are indicated  by solid circles. The codespace is defined as a common eigenspace of
site and plaquette stabilizers --- products of Pauli $X$ and $Z$ over 
stars (red) and plaquettes (green).  Logical Pauli operators $\overline{Z}$ and $\overline{X}$
correspond to a loop encircling one of the defects and the path connecting the two defects respectively.}
\end{figure}

Let us now describe our numerical results.  We consider the standard surface
code on a square lattice~\cite{Kitaev03,BK:surface} with a pair of defects
representing one logical qubit, see Fig.~\ref{fig:flat}.  A defect is defined
as a square block of plaquettes removed from the lattice.  Logical Pauli
operators $\overline{Z}$ and $\overline{X}$ are products of Pauli $Z$ and $X$
over a loop encircling one of the defects and over a path on the dual lattice
connecting the two defects respectively, see Fig.~\ref{fig:flat}.  The
dimensions of the lattice in our simulations were chosen such that the code
distance is $d=4r$, where $r$ is the linear size of the defects. Accordingly,
each defect is separated from the external boundary of the lattice and from the
other defect by distance $4r$, see Fig.~\ref{fig:decoding}.  Simulations were
performed only for $r\ge 2$ to avoid finite size effects.

We first analyze a toy error model with a noiseless syndrome readout where
every qubit is subjected to independent bit-flip and phase flip-errors, each applied
with probability $p$ (qubits on the external boundary of the lattice may
be treated specially, see Section~\ref{sec:implement} for details).  Once all
errors have been generated, one computes the error syndrome and finds the most
likely recovery operator consistent with the syndrome, see
Sections~\ref{sec:SC},\ref{sec:EC}.  A logical loop-like or path-like error
occurs if the recovery operator differs from the actual error by a logical
operator $\overline{Z}$ or $\overline{X}$ (modulo stabilizers).  We are
interested in the logical error probability $P_L=P_L(p,r)$ for both loop-like
and path-like logical errors.  In the chosen geometry the code distance
$d$ can only be a multiple of four; thus, it will be more convenient to define the
decay rate $\alpha$ such that 
\begin{equation}
\label{alpha}
P_L(p,r)\sim \exp{[-\alpha(p) r]}.
\end{equation}
It differs from the decay rate defined earlier by a factor of $4$.  The largest
error rate $p_{th}$ such that $\alpha(p)>0$ for all $p<p_{th}$ is known as the
error threshold. For the chosen model $p_{th}\approx 10.03\%$,
see~\cite{WHP2002}.  Our simulations were performed only for error rates $p\le
8\%$ to avoid finite-size effects which become important for $p\approx
p_{th}$.  The numerical results are presented in Fig.~\ref{fig:flat2345}.  The
logical error probability $P_L$ is shown as a function of $p$ for a few
small values of $r$.  This provides enough data to perform the exponential
fitting $P_L\sim \exp{(-\alpha r)}$ and extract the decay rate $\alpha$. The
plot of $\alpha$ as a function of $p$ is shown in Fig.~\ref{fig:flat_alpha}.
This figure also shows the estimate of $\alpha$ based on the asymptotic
low-$p$ formula $P_L\approx A_d\,  p^{d/2}$. It is clear that the latter
provides a poor approximation of $\alpha$ for large and moderately small error
rates, especially in the case of path-like errors.  Since $\alpha(p_{th})=0$,
one can easily extend the function $\alpha(p)$ to the interval $8\%\le p\le
p_{th}$ using a linear interpolation.  More details on the simulation methods
can be found in Section~\ref{sec:implement}.

\begin{figure}
\includegraphics[scale=0.65]{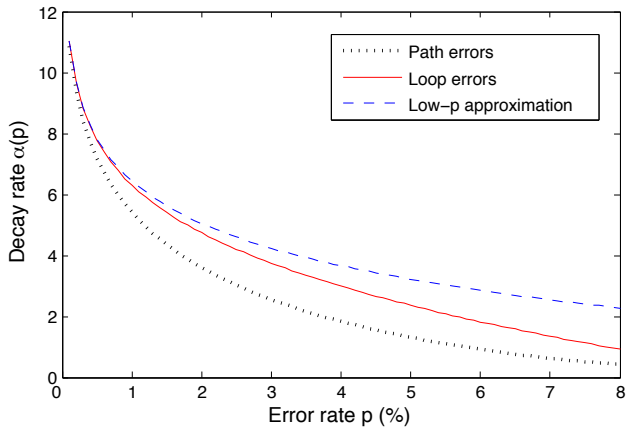}
\caption{\label{fig:flat_alpha}
(Color online) {\em Noiseless syndrome readout.}  The decay rate  $\alpha(p)$ 
in the exponential scaling  $P_L\sim \exp{[-\alpha(p) r]}$
for loop-like and path-like logical errors computed by
the upward splitting method. 
The dashed line shows the
estimate of $\alpha(p)$ based on the asymptotic low-$p$
formula $P_L\sim {4r \choose 2r} p^{2r}$, that is,
$\alpha(p)=-4\log{2} - 2\log{(p)}$.
}
\end{figure}

\begin{figure}[h]
\includegraphics[scale=0.65]{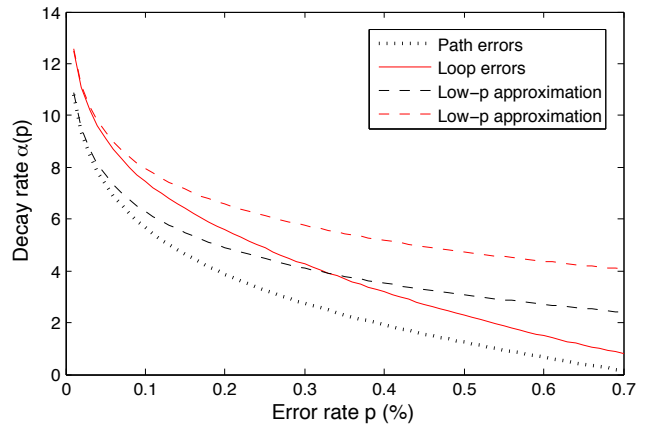}
\caption{\label{fig:alpha}
(Color online)
{\em Noisy syndrome readout.}  The decay rate  $\alpha(p)$ 
for loop-like and path-like logical errors computed by
the downward splitting method. 
The dashed lines show the estimate of $\alpha(p)$ based on
the low-$p$ asymptotic formula $P_L\sim A_r\, p^{2r}$,
that is, $\alpha(p)=-c-2\log{p}$, see  the last line of Eq.~(\ref{Fit}).}
\end{figure}

The following heuristic ansatz was used to fit the numerical data for $P_L$:
\begin{eqnarray}
P_L(p,r) &\approx& \exp{[-\alpha(p)(r-r_0)]} P_L(p,r_0),\label{Fit} \\
P_L(p,r_0) &=& \exp{[ 2r_0 \log{(p)} + x(p)]},\nonumber \\
-\alpha(p)&=&c + 2\log{(p)} + \log{( 1 + y(p))}. \nonumber
\end{eqnarray}
Here $r_0=2$ is the smallest defect size in our simulations, $c$ is a constant
coefficient, and the functions $x(p)$, $y(p)$ are low-degree polynomials such
that $y(0)=0$.  This ansatz is a slightly refined version of the exponential
scaling Eq.~(\ref{alpha}) that attempts to reproduce the pre-exponential
factor.  We fit $x(p)$ and $y(p)$ by the second and the third degree
polynomials,
\begin{equation}
\label{xy}
x(p)=\sum_{n=0}^2 x_n p^n, \quad y(p)=\sum_{n=1}^3 y_n x^n.
\end{equation}
The terms proportional to $\log{(p)}$ in Eq.~(\ref{Fit}) are chosen to
reproduce the asymptotic formula $P_L\sim A_d \, p^{d/2}$ in the limit
$p\to 0$.
The coefficients $c$, $x_i$ and $y_i$ found by fitting the numerical data
are listed in Table~\ref{table:fit}. 
 \begin{table}[h]
\begin{ruledtabular}
\begin{tabular}{c | c | c | c | c }
& \multicolumn{2}{c|}{Noiseless syndromes} & \multicolumn{2}{c}{Noisy syndromes} \\
\hline
& Loop errors & Path errors & Loop errors & Path errors \\
\hline
  $x_0$ & $3.18$ &$4.25$ &	$8.67$ &   $9.88$        \\
    $x_1$& $24.2$ & $64.6$  &	$405$	      &         $1.17\times 10^3$      \\
  $x_2$ &$19$ & $-347$&	$5.83\times 10^3$	      &                $-7.64\times 10^4$   \\
  $c$ & $4\log{2}$ & $4\log{2}$ & $5.86$		      &    $7.52$               \\
  $y_1$&$12.2$ &  $193$ & $551$		      &           $880$         \\
 $y_2$& $297$ &$-1570$  & $3.27\times 10^4$		      &             $4.69\times 10^3$      \\
   $y_3$& $0$ & $0$  & $5.72\times 10^7$		      &             $6.04\times 10^6$   \\
\end{tabular}
\caption{Coefficients  in the  fitting formulas 
Eqs.~(\ref{Fit},\ref{xy}).\label{table:fit}}
\end{ruledtabular}
\end{table}

\begin{figure}[htb]
\includegraphics[scale=0.4]{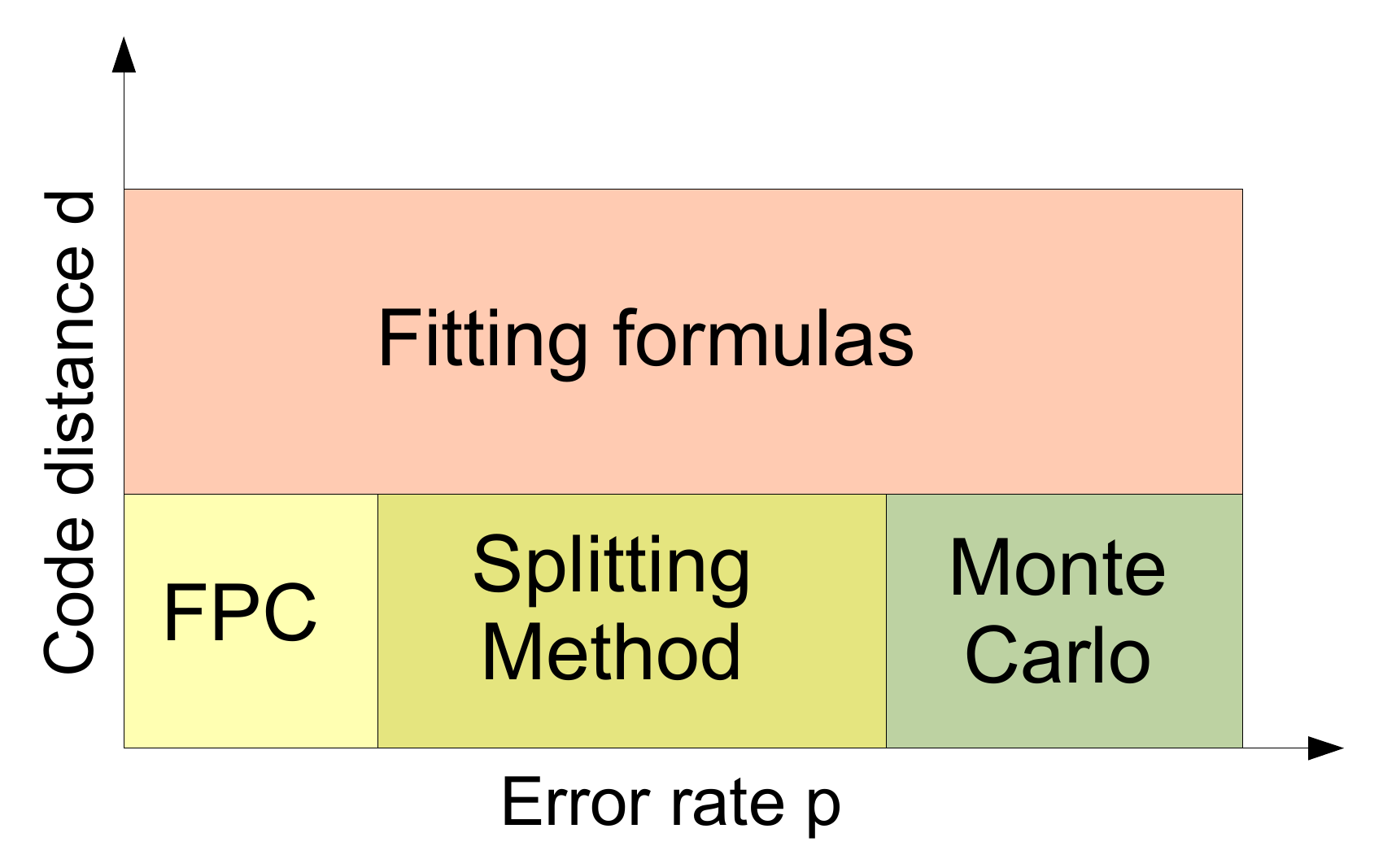}
\caption{Four methods of estimating the logical error probability $P_L$.
For small code distances $d$ one can compute $P_L$ numerically
by the Monte Carlo simulation (if $p$ is close to the threshold) or  
by the splitting method (if $p$ is sufficiently below the threshold).
The fault-path counting (FPC) is applicable for very small error rates. 
Numerical data are used to estimate the decay rate $\alpha$
in the fitting formula  $P_L\sim \exp{[-\alpha(p) d]}$ 
which is applicable for larger values of $d$. 
\label{fig:diagram}}
\end{figure}

Next we analyze a more realistic noise model with a noisy syndrome readout.
Here the syndromes of site and plaquette stabilizers are measured by applying a
sequence of CNOT gates which couple the respective code qubits with ancillary
qubits collecting the syndrome information, see Section~\ref{sec:noisy}.
Independent depolarizing errors occur with probability $p$ on each qubit
preparation, measurement, and CNOT gate.  Our noise model mostly coincides with the
one introduced by Fowler in  Ref.~\cite{Fowler08}. To enable reliable error
correction, the syndrome readout is repeated $d$ times, where $d=4r$ is the
code distance.  The lattice geometry is time-independent which corresponds to
the logical identity gate (fault-tolerant storage).  Once all syndromes have
been generated, the most likely configuration of errors consistent with the
syndrome and the corresponding recovery operator on code qubits are computed,
see Section~\ref{sec:EC}.  As before, a logical loop-like or path-like error
occurs if the recovery operator differs from the actual accumulated error on
code qubits by a logical operator $\overline{Z}$ or $\overline{X}$ (modulo
stabilizers).  We are interested in the logical error probability
$P_L=P_L(p,r)$ for both loop-like and path-like logical errors.  More
precisely, we define $P_L$ as the logical error probability divided by the
total number of syndrome readout steps, which is $4r$ in our simulations.
Estimates of the error threshold $p_{th}$ for a noisy syndrome readout vary
between $0.75\%$ and $0.9\%$, see~\cite{RHG07,Fowler12,fowler2013accurate}.  To
avoid finite size effects, simulations were performed only for $p\le 0.7\%$ (as
before, one can use linear interpolation to cover the interval between $0.7\%$
and $p_{th}$).  Our numerical results are presented in
Figs.~\ref{fig:alpha},\ref{fig:234}.  The fitting curves in Fig.~\ref{fig:234}
represent the ansatz Eqs.~(\ref{Fit},\ref{xy}) where the coefficients $c$,
$x_i$ and $y_i$ are defined in Table~\ref{table:fit}.  As before, the
numerically computed decay rate $\alpha$ is compared with its estimate based on
the low-$p$ asymptotic formula $P_L=A_d\, p^{d/2}$.  The coefficient $A_d$ was
found by setting $x(p)=x(0)$ and $y(p)=0$ in Eqs.~(\ref{Fit}).  We observe that
the asymptotic low-$p$ formula significantly overestimates $\alpha$ for large
and moderately small error rates.  More details on the simulation methods can
be found in Section~\ref{sec:noisy}.

The limiting factor in our simulations was the running time of the Metropolis
subroutine which grows rapidly as one increases the error rate, see
Section~\ref{sec:implement} for details. Accordingly, we were able to implement
the splitting method only for sufficiently small errors rates, $p\le p^*$,
where $p^*\approx p_{th}/2$.  Fortunately, the cutoff error rate $p^*$ was
large enough to enable Monte Carlo computation of $P_L(p)$ for all $p\ge p^*$.
In the case of noisy syndrome readout, we used downward splitting to access
error rates $p\le p^*$. For noiseless syndrome readout, where asymptotic
low-$p$ formulas for $P_L(p)$ are readily available, we used upward splitting
to access error rates $p\le p^*$.  Accordingly, in the latter case we were able
to test the correctness of the splitting method by comparing the values of
$P_L(p^*)$ found independently by the splitting and the Monte Carlo methods,
see Fig.~\ref{fig:flat2345}.  For noisy syndrome readout, the correctness of the
splitting method is partially confirmed by the fact that the function $P_L(p)$
computed by the two methods has the same derivative on both sides of $p^*$, see
Fig.~\ref{fig:234}.  A diagram illustrating the applicability region of each
simulation method is shown in Fig.~\ref{fig:diagram}.
 
We hope that the new simulation techniques and the fitting formulas can find
applications in the estimation of the fault-tolerance overhead for quantum
computing architectures based on the surface code.  Significant progress in
this direction has been recently made in
Refs.~\cite{RHG07,ghosh2012surface,fowler2012surface}.  We hope that our
results can refine estimates made in
Refs.~\cite{ghosh2012surface,fowler2012surface} by replacing the asymptotic
low-$p$ formulas for the logical error probability with more accurate
approximations, such as the one defined in Eqs.~(\ref{Fit},\ref{xy}) and
Table~\ref{table:fit}.  We also hope that the fitting formulas
Eqs.~(\ref{Fit},\ref{xy}) have applicability beyond the toy noise models
considered in this paper. 
 
In the rest of this paper we introduce the general idea of the splitting
method, see Section~\ref{sec:split}, and specialize it to the surface code
settings, see Sections~\ref{sec:SC}-\ref{sec:noisy}.  It should be emphasized
that at present the splitting method is a heuristic algorithm. Deriving
rigorous bounds on its running time and the approximation error goes beyond the
scope of this paper.  Nevertheless, we make first steps in this direction in
Section~\ref{sec:corr} (see Lemmas~1-4).  Some of our results such as the
maximum weight matching decoding algorithm presented in Section~\ref{sec:EC}
and a decoder-independent definition of correctability, see
Section~\ref{sec:corr}, might be interesting on their own right. 
We conclude by discussing some open problems in Section~\ref{sec:concl}.

 \begin{figure*}[htb]
\includegraphics[scale=0.6]{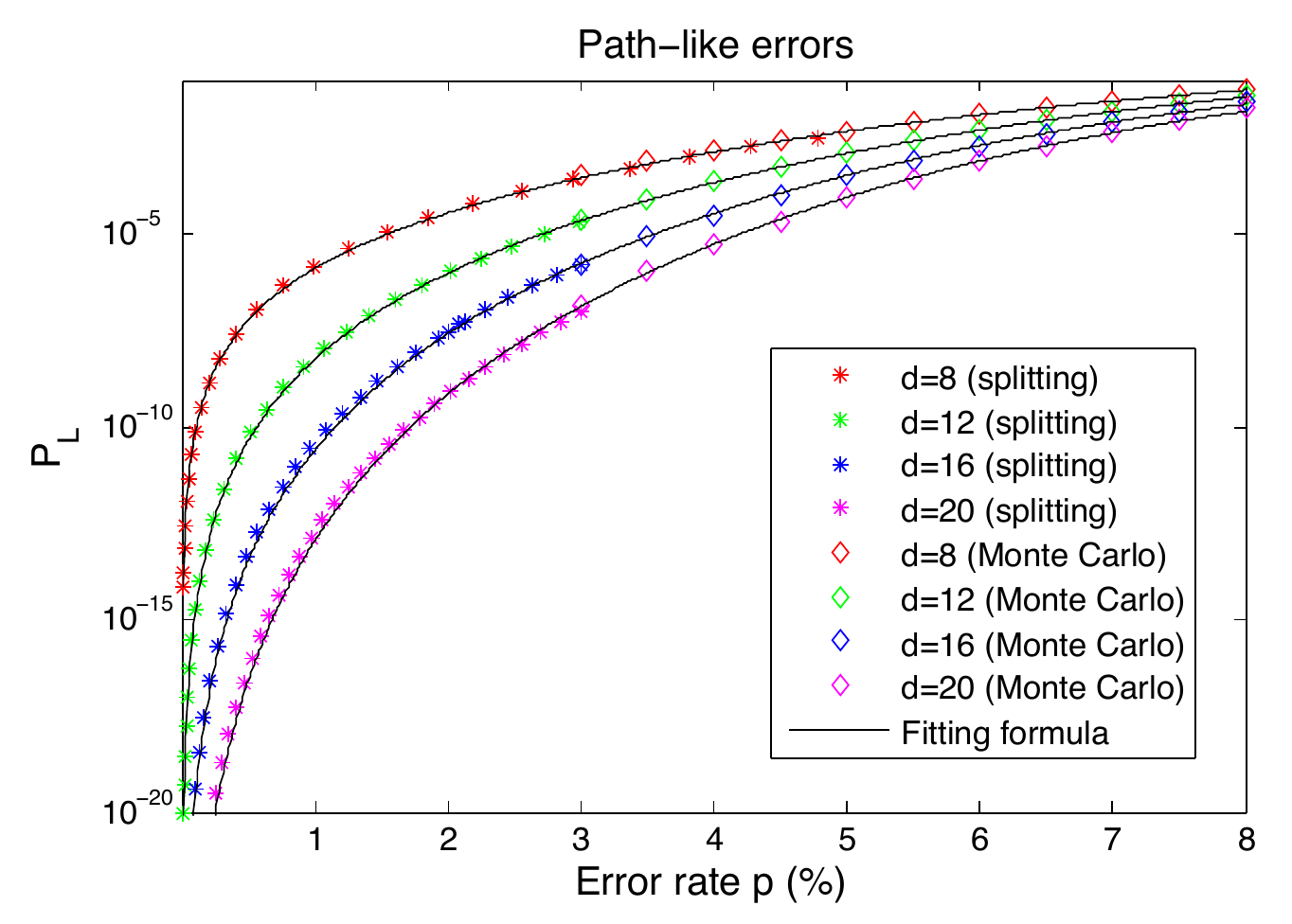}
\includegraphics[height=59mm]{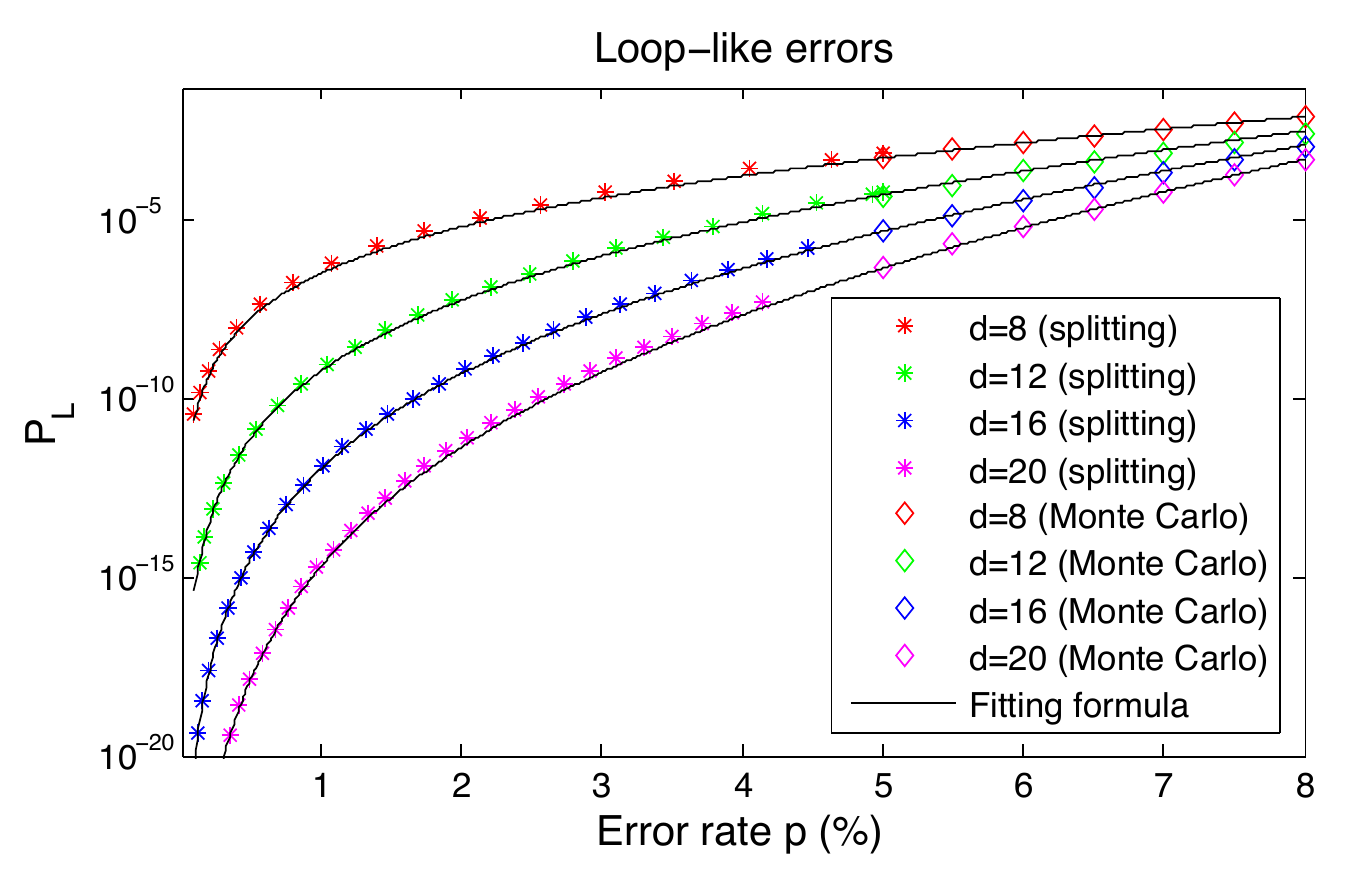}
\caption{\label{fig:flat2345}  (Color online) {\em Noiseless syndrome readout.}
Independent bit-flip and phase-flip errors occur with probability $p$ on each physical qubit.
Logical error probability  $P_L(p)$ has been  computed by the upward splitting method
for error rates $p\le p^*$ and by Monte Carlo method for $p\ge p^*$.
Here $p^*= 3\%$  for path-like errors and $p^*=5\%$ for loop-like errors. 
The splitting sequence starts at error rate $p_1=0.1\%$
where one can use asymptotic formulas
$P_L\approx 0.5 r {4r \choose 2r} p^{2r}$ for path-like errors
and $P_L\approx 0.5 {4r \choose 2r} p^{2r}$ for loop-like errors.
Here $r=d/4$ is the linear size of the defects. 
A good agreement between the values of $P_L(p)$ obtained by the
two methods is observed at the junction point $p=p^*$. 
Solid lines represent fitting
formulas Eqs.~(\ref{Fit},\ref{xy}).
}
\end{figure*}

\begin{figure*}[htb]
\includegraphics[scale=0.6]{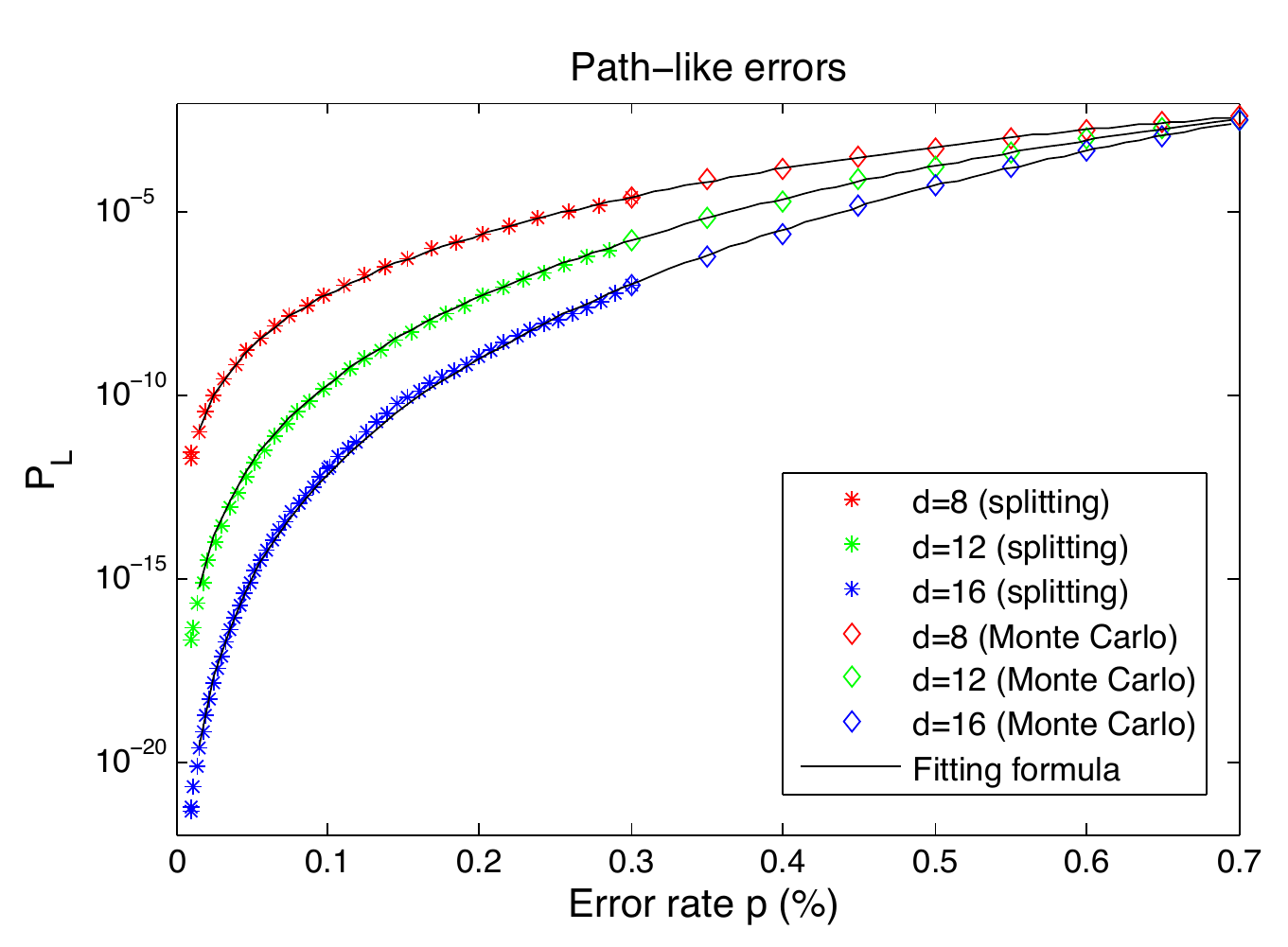}
\includegraphics[scale=0.6]{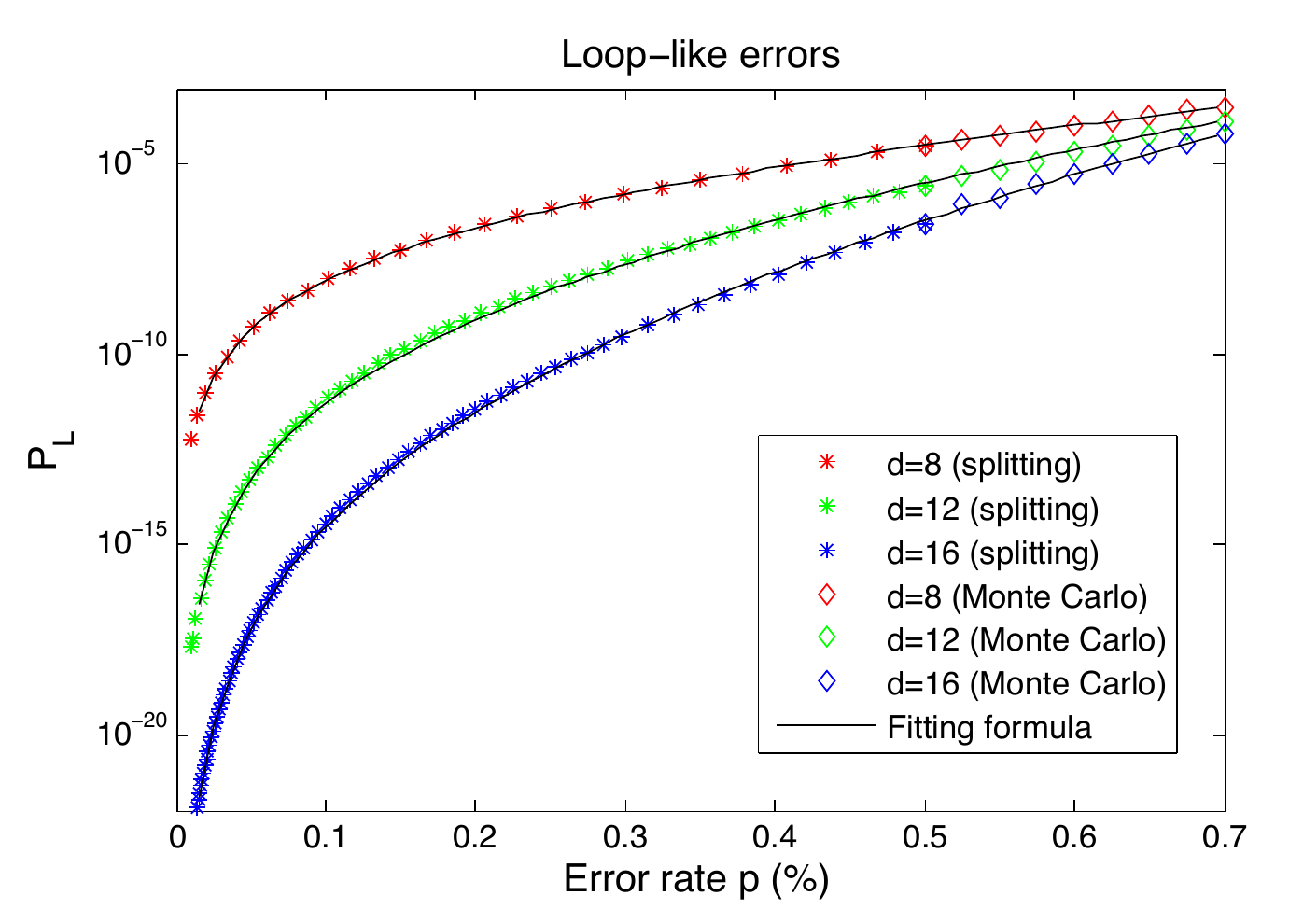}
\caption{\label{fig:234}
(Color online) {\em Noisy syndrome readout.}
Independent depolarizing errors occur with probability $p$ on each
qubit  preparation, measurement, and CNOT gate.
Logical error probability per time step $P_L$ has been computed by the downward splitting method
for error rates $p\le p^*$ and by Monte Carlo method for $p\ge p^*$.
Here $p^*= 0.3\%$  for path-like errors and $p^*=0.5\%$ for loop-like errors. 
The splitting sequence starts at $p_1=p^*$.
A good agreement between the derivatives of $P_L(p)$
computed by the two methods  is observed at the junction point $p=p^*$. 
Solid  lines represent fitting
formulas Eqs.~(\ref{Fit},\ref{xy}).
Generating the data presented in Figs.~\ref{fig:flat2345},\ref{fig:234}
took roughly 20,000 CPU hours on PowerPC~450 processors.
}
\end{figure*}

\section{Rare event simulation}
\label{sec:split}

This section provides some necessary background on the acceptance ratio method
due to Bennett~\cite{bennett1976} and the splitting method.  Let $\Omega$ be a
finite set of events, $\calF\subseteq \Omega$ be a subset of failure events,
and $\pi$ be a normalized probability distribution that assigns a probability
$\pi(E)$ to any event $E\in \Omega$.  Our goal is to calculate the overall
failure probability 
\[
\pi(\calF)\equiv \sum_{E\in \calF} \pi(E).
\]
For any probability distributions $\pi_1,\pi_2,\ldots,\pi_t$ such that
$\pi_t\equiv \pi$ and such that $\pi_1(\calF)$ is known, one can use the
obvious identity
\begin{equation}
\label{SM1}
\pi(\calF)=\pi_1(\calF) \prod_{j=1}^{t-1} \frac{\pi_{j+1}(\calF)}{\pi_{j}(\calF)}.
\end{equation}
The splitting method attempts to compute $\pi(\calF)$ by evaluating
each ratio $\pi_{j+1}(\calF)/\pi_j(\calF)$ in Eq.~(\ref{SM1}) separately.
The method is applicable whenever 
the distributions $\pi_1,\ldots,\pi_t$ have the following properties:\\
(i) For any event $E\in \Omega$ the probability $\pi_j(E)$ can be  computed efficiently.\\
(ii) There exists
an efficient randomized algorithm $\calM_j$ that generates samples $E\in \calF$
drawn from the conditional probability distribution 
$\pi_j(E|\calF) \equiv \pi_j(E)/\pi_j(\calF)$.\\
(iii) The conditional distributions $\pi_j(E|\calF)$ and $\pi_{j+1}(E|\calF)$ have a non-negligible overlap,
as quantified below.

Given any function $f\, : \, \Omega\to \RR$, we shall use
the shorthand notation 
\begin{equation}
\label{SM2}
\la f\ra_j = \frac1{\pi_j(\calF)} \sum_{E\in \calF} \pi_j(E) f(E)
\equiv \sum_{E\in \calF} \pi_j(E|\calF) f(E).
\end{equation}
Then one can easily check that
\begin{equation}
\label{SM3}
\frac{\pi_{j+1}(\calF)}{\pi_{j}(\calF)}= C \frac{\la g(C\pi_j/\pi_{j+1})\ra_j}{\la g(C^{-1}\pi_{j+1}/\pi_j)\ra_{j+1}}
\end{equation}
for any constant $C>0$ and any function $g\, : \, \RR \to \RR$ satisfying the ``detailed balance" condition
\begin{equation}
\label{SM4}
g(x)=x^{-1}g(x^{-1}).
\end{equation}
The expectation values  needed to compute the righthand side Eq.~(\ref{SM3}) can be approximated 
by calling the algorithm $\calM_j$ to generate sufficiently many samples
$E_1,\ldots,E_N\in \calF$ drawn from the distribution $\pi_j(E|\calF)$ and using
an estimate
\begin{equation}
\label{SM5}
\la g(C^{\pm 1}\pi_j/\pi_{j\pm 1})\ra_j \approx \frac1N \sum_{\alpha=1}^N g(C^{\pm 1}\pi_j(E_\alpha)/\pi_{j\pm 1}(E_\alpha)).
\end{equation}
As was shown in~\cite{bennett1976}, for a fixed 
number of samples $N$  the statistical error in Eq.~(\ref{SM3}) 
is minimized if one chooses
\begin{equation}
\label{SM6}
g(x)=\frac1{1+x}
\end{equation}
while the constant $C$ satisfies 
\begin{equation}
\label{SM7}
 \la g(C\pi_j/\pi_{j+1})\ra_j =\la g(C^{-1}\pi_{j+1}/\pi_j)\ra_{j+1}.
\end{equation}
In practice, both sides of Eq.~(\ref{SM7}) are replaced by
their $N$-sample approximations (as defined in Eq.~(\ref{SM5}));
the resulting equation can then be solved for $C$ to obtain the optimal value.
This choice of $g(x)$ and $C$ guarantees that the ratio $\pi_{j+1}(\calF)/\pi_j(\calF)$
is estimated with a relative error $\sigma_j$, where
\begin{equation}
\label{SM8}
\sigma_j^2 = \frac2N \left[ \left( \sum_{E\in \calF} \frac{2\pi_j(E|\calF) \pi_{j+1}(E|\calF)}{\pi_j(E|\calF) +
\pi_{j+1}(E|\calF)} \right)^{-1} -1\right],
\end{equation}
see Ref.~\cite{bennett1976}.
Simple algebra shows that 
\begin{equation}
\label{SM9}
\sigma_j^2\le \frac2{N(1-\| \pi_j(\cdot |\calF)- \pi_{j+1}(\cdot|\calF)\|_1)},
\end{equation}
where $\|p-q\|_1\equiv (1/2)\sum_i |p_i-q_i|$ is the total variation distance.
In particular, if $\| \pi_j(\cdot |\calF)- \pi_{j+1}(\cdot|\calF)\|_1\le 1-\delta$ for
all $j=1,\ldots,t-1$ then the overall failure probability $\pi(\calF)$ is estimated
with a relative error 
\begin{equation}
\label{SM9a}
\sigma\sim \frac{t}{\sqrt{\delta N} }.
\end{equation}
(In fact, since different terms in Eq.~(\ref{SM1}) are computed independently, one may expect that
the statistical errors accumulate in a random walk fashion).

\section{Surface codes}
\label{sec:SC}
 
We consider $n$ physical qubits located on the edges of a square lattice
$\Sigma$ with open boundary conditions, possibly containing one or several
defects, see Fig.~\ref{fig:flat} for an example.  A defect is defined as a
square block of plaquettes removed from the lattice.  The sets of sites, edges,
and plaquettes of the lattice are denoted $\Sigma_0$, $\Sigma_1$, and
$\Sigma_2$ respectively.  Given a subset of qubits $E\subseteq \Sigma_1$, let
$X(E)=\prod_{e\in E} X_e$ and $Z(E)=\prod_{e\in E} Z_e$, where $X_e$ and $Z_e$
are the Pauli  operators $\sigma^x$ and $\sigma^z$ acting on the qubit $e$.
For any site $u\in \Sigma_0$ and any
plaquette $f\in \Sigma_2$ let $A_u\subseteq \Sigma_1$ be the set of edges
incident to $u$ and $B_f\subseteq \Sigma_1$ be the set of edges lying on the
boundary of $f$.  Pauli operators $X(A_u)$ and $Z(B_f)$ are called {\em
stabilizers} of the surface code.  The stabilizers pairwise commute and have a
common invariant subspace 
\[
\calL=\{\psi \in (\CC^2)^{\otimes n} \, : \, X(A_u)\psi=Z(B_f)\psi =\psi \; \;
\forall u,f\}.
\]
The subspace $\calL$ is used to encode logical qubits.  It is well-known that
$\dim{(\calL)}=2^{b_1}$, where $b_1$ is the first Betti number of the lattice.
For open boundary conditions $b_1$ coincides with the number of defects.  Below
we only consider single and double defect geometries, $b_1=1,2$. Given a Pauli
error $P$, the syndrome of $P$ is defined as the set of all stabilizers
anti-commuting with $P$. 

To describe the relationship between single-qubit
errors and the corresponding syndromes it will be convenient to introduce
a {\em  decoding graph} $\calG=(\calV,\calE)$. There will be one decoding graph
for each type of Pauli errors (bit-flip and phase-flip errors).  Given a subset of edges $E\subseteq \calE$,
let $\partial E\subseteq \calV$ be the set of vertices $u\in V$ having odd number
of incident edges from $E$.

Let us start with phase  errors.  In this case, the decoding graph
$\calG=(\calV,\calE)$ is identical to the physical lattice, that is,
$\calV=\Sigma_0$ and $\calE=\Sigma_1$. By construction, a phase-flip error on
any edge $e$ of $\calG$ creates a pair of syndromes at the two end-points of
$e$.  More generally, given a subset of edges $E\subseteq \calE$, the syndrome
of the Pauli error $Z(E)$ coincides with the set $\partial E$.

Let us now consider bit-flip errors. In this case the decoding graph
$\calG=(\calV,\calE)$ loosely coincides with the dual of the physical lattice.
More precisely, let $T\subseteq \Sigma_1$ be the set of all qubits lying on the
boundary of the lattice (either the boundary of some defect or the external
boundary).  We choose $\calV=\Sigma_2\cup T$ and $\calE=\Sigma_1$.  A bit-flip
error on any edge $e\in \Sigma_1\backslash T$ creates a pair of syndromes at
the two plaquettes  $f,f'$ adjacent to $e$.  In the decoding graph $f,f'$ are
the two end-points of $e$.  A bit-flip error on a boundary edge $e\in T$, however,
creates a syndrome only at one plaquette $f$ adjacent to $e$. On the decoding
graph $e$ is a ``hanging edge'' connecting vertex $f$ with a degree-$1$ vertex
labeled by $e$ itself.  More generally, given a subset of edges $E\subseteq
\calE$, the syndrome of a Pauli error $X(E)$ is $(\partial E)\backslash T$.  To
deal with both types of errors on the same basis, we set $T=\emptyset$ in
the case of phase-flip errors. 

\begin{figure}[htb]
\centerline{\includegraphics[height=6cm]{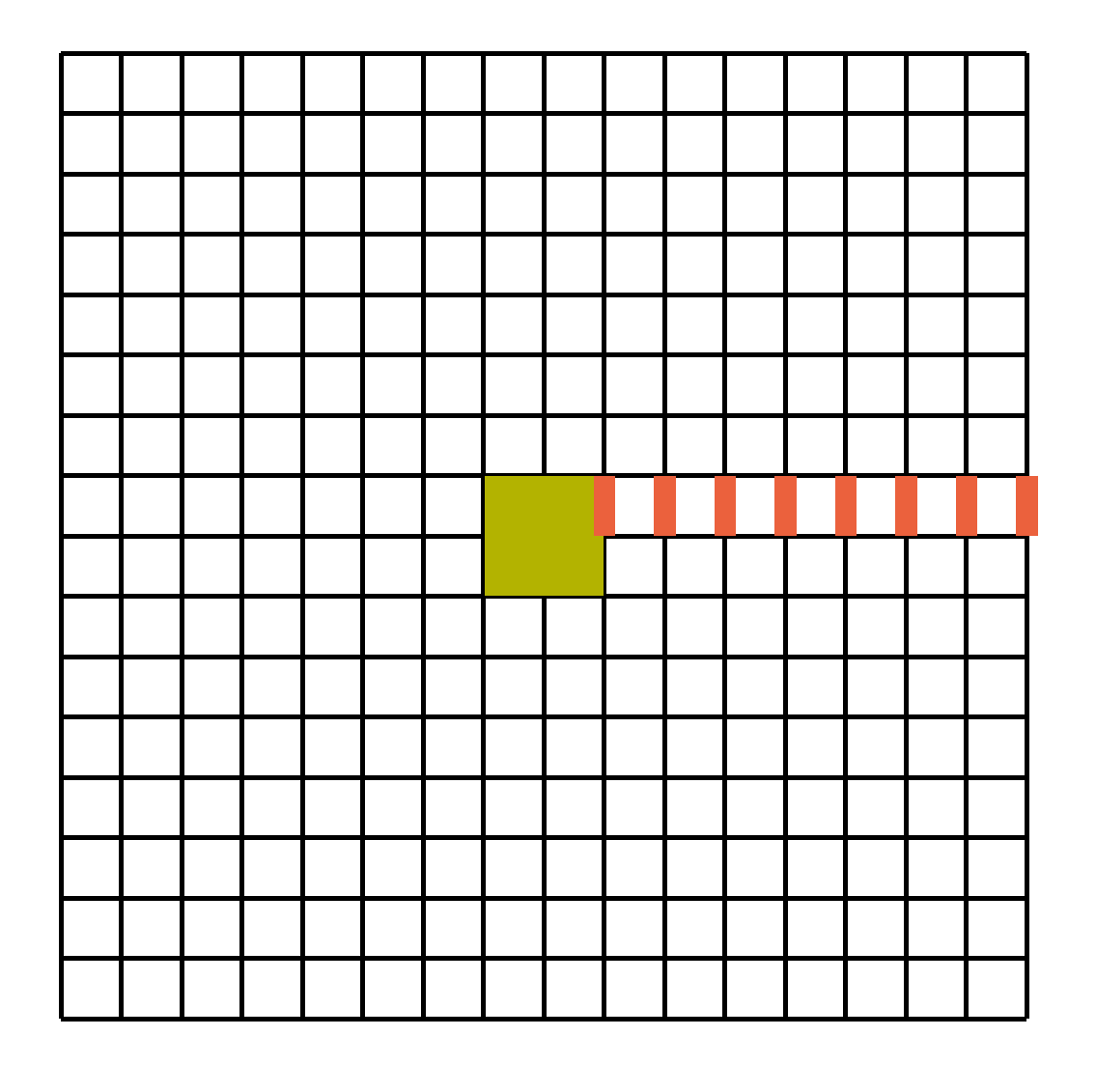}}
\centerline{\includegraphics[height=6cm]{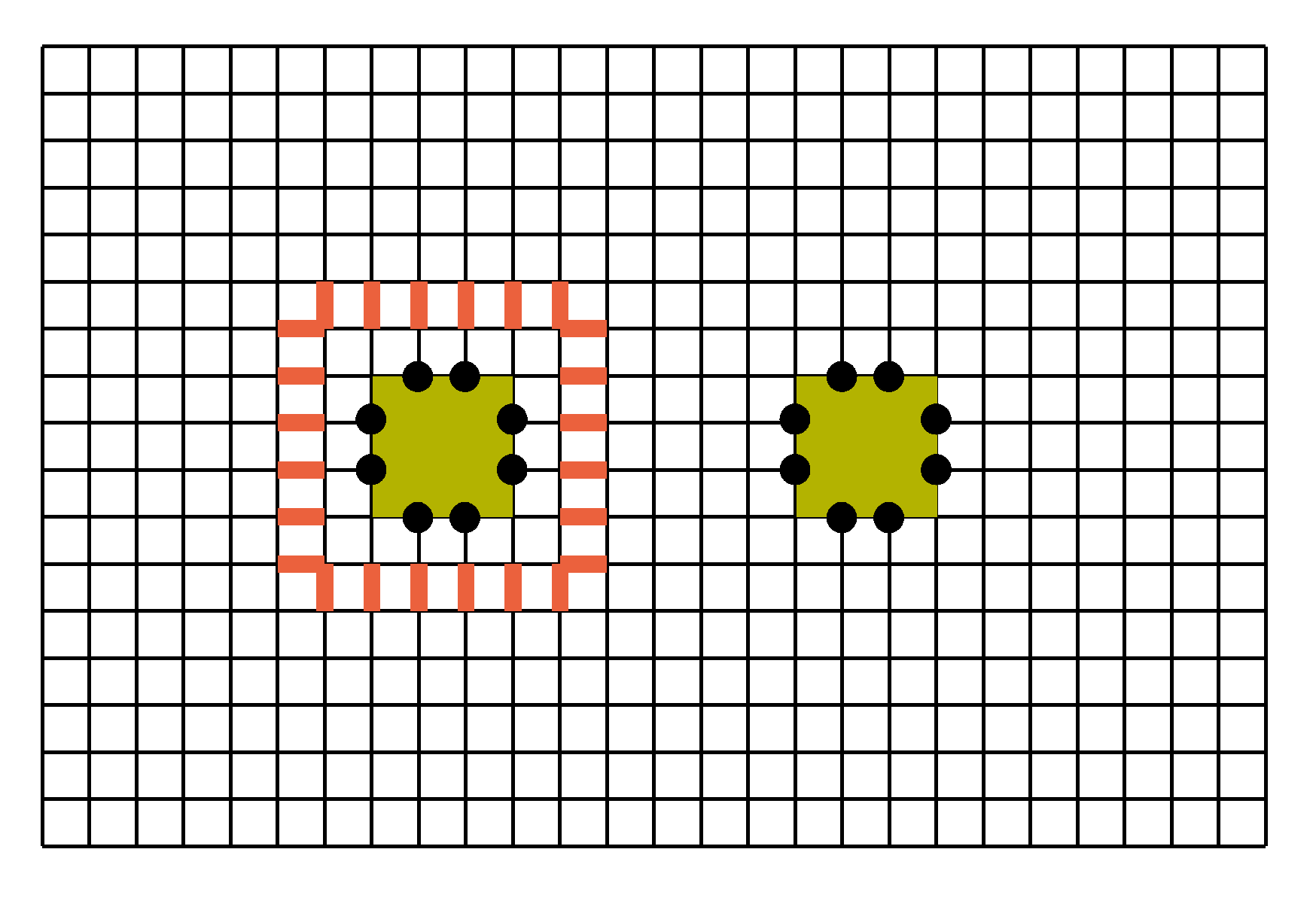}}
\caption{(Color online)
Examples of the decoding graph $\calG=(\calV,\calE)$ used in the simulations of
loop-like errors (top) and path-like errors (bottom) for $r=2$.
Logical chains $\Gamma$ are highlighted in red. 
In the case of path-like errors the graph contains 
``hanging edges" $(u,v)\in \calE$ such that $v\in T$ is a degree-$1$ vertex
that carries no syndrome. To avoid clutter, a hanging edge $(u,v)$
 is represented by a solid circle centered at $u$ (the degree-$1$ vertex $v$ is not shown).
The geometry depends on three parameters $r,s,b$, where
$r$ is the linear size of the defects, $s$ is the separation between the defects,
and $b$ is the length of the buffer zone separating the defects
and the external boundary of the lattice.
The simulations were performed for $s=b=4r$.
 This guarantees that the  code distance is $d=4r$.
\label{fig:decoding}}
\end{figure}

Let $\psi\in \calL$ be a logical state and $P$ be an unknown Pauli error.  We
shall deal with phase-flip and bit-flip errors independently, so without loss
of generality we can assume that
$P=Z(E)$ or $P=X(E)$ for some subset of edges $E\in \calE$ in the
corresponding decoding graph $\calG=(\calV,\calE)$.  Following
Ref.~\cite{Dennis01} we shall refer to the subset $E$ as an {\em error chain}.
Error correction is a two-stage process that consists of a syndrome readout
followed by a recovery step.  A syndrome readout takes as input the corrupted
state $P\, \psi$ and performs a non-destructive eigenvalue measurement of each
stabilizer. It determines the syndrome  $S=\partial E\backslash T$.  We shall
first consider the case when the syndrome readout circuit contains no errors.
Noisy syndrome readout is discussed in Section~\ref{sec:noisy}.  A recovery
step is determined by a decoding algorithm that takes as input the measured
syndrome $S\subseteq \calV\backslash T$ and returns a {\em recovery chain}
$R\subseteq \calE$.  The corrupted state $P\psi$ is then acted upon by a
recovery operator $P'=Z(R)$ for phase-flip correction
or $P'=X(R)$ for bit-flip correction.  Below we shall only consider
decoders that return any corrupted state to the logical subspace $\calL$.
Equivalently, the recovery chain and the error chain must have the same
syndrome: $\partial R\backslash T=\partial E\backslash T$.  This condition
guarantees that $P'P$ commutes with all stabilizers and thus $P'P\psi\in \calL$
for any $\psi\in \calL$. 
Error correction is successful if the restriction of $P'P$ onto $\calL$ is the identity operator. 
Otherwise, error correction results in a logical error. 
Below we only consider encodings of a single logical qubit with logical
Pauli operators $\overline{X}=X(\Gamma^X)$, $\overline{Z}=Z(\Gamma^Z)$.
Here $\Gamma^Z$ is a closed loop on the primal lattice $\Sigma$ encircling
the (left-most) defect and $\Gamma^X$ is a path on the lattice dual to $\Sigma$
connecting the defect with the external boundary (for the double-defect geometry
$\Gamma^X$ connects the two defects with each other). One can easily check
that $\overline{X}$, $\overline{Z}$ commute with all stabilizers and anti-commute with each other. 
The decoding graph will be equipped with a {\em logical chain } $\Gamma\subseteq \calE=\Sigma_1$
such that $\Gamma=\Gamma^X$ for phase-flip errors and $\Gamma=\Gamma^Z$ for bit-flip errors. 
Error correction is  successful  iff $P'P$ commutes with the logical operators $\overline{X}$, $\overline{Z}$.
Equivalently, $R\oplus E$ must have even overlap with $\Gamma$. 
 The decoding graphs used to generated the
data shown on Fig.~\ref{fig:flat2345} for $r=2$ and the corresponding logical chains
$\Gamma$ are shown on Fig.~\ref{fig:decoding}.

\section{Decoding algorithms}
\label{sec:EC}

Here we discuss algorithms for choosing a recovery chain.
This material is mostly based
on Ref.~\cite{Wang11}.

Let $\calG=(\calV,\calE)$ be the decoding graph defined above.
We assume that every edge $e\in \calE$ has some specified error rate $0\le p(e)\le 1/2$.
We shall only consider noise models such that errors on different edges of $\calG$
are independent. 
An error chain $E\subseteq \calE$ then appears with a probability
\begin{equation}
\label{probE}
\pi(E)=\prod_{e\in E} p(e)\prod_{e\in \calE\backslash E} (1-p(e)).
\end{equation}
Introduce edge weights $\phi(e)\ge 0$ such that 
\[
\exp{[-\phi(e)]}=\frac{p(e)}{1-p(e)}.
\]
Then $\pi(E)=c\cdot \exp{[-\phi(E)]}$, where
\begin{equation}
\label{phi}
\phi(E)=\sum_{e\in E} \phi(e),
\end{equation}
and $c$ is a constant coefficient independent of $E$.
We shall refer to $\phi(E)$ as a {\em weight} of $E$. 
Following Refs.~\cite{Dennis01,Wang11} we shall choose a recovery
chain $R$ as the most likely error chain consistent with the observed syndrome $S$. 
Equivalently, a candidate recovery chain $R$ must obey $(\partial R)\backslash T=S$ and 
$\phi(R)\le \phi(R')$ for any chain $R'\subseteq \calE$ satisfying $(\partial R')\backslash T=S$.
Thus decoding amounts to solving one of the following problems.
\begin{problem}
Given a graph $\calG=(\calV,\calE)$ with non-negative edge weights $\phi(e)$
and a subset of vertices $S\subseteq \calV$. Find a
minimum weight chain $E\subseteq \calE$ satisfying $\partial E=S$.
\end{problem}
\begin{problem}
Given a graph $\calG=(\calV,\calE)$ with non-negative edge weights $\phi(e)$
and disjoint sets $S,T\subseteq \calV$. Find a 
minimum weight chain $E\subseteq \calE$ satisfying $(\partial E)\backslash T=S$.
\end{problem}
In combinatorial optimization, Problem~1 is known under the name
minimum weight $T$-join (in our case $T\equiv S$), see Ref.~\cite{KorteBook}.
It can be solved in time $O(|\calV|^3)$ 
for any graph $\calG$ and any real edge weights by 
a reduction to the minimum weight perfect matching problem
and solving the latter using Edmonds's blossom 
algorithm~\cite{edmonds1965paths} or its subsequent improvements,
see Ref.~\cite{KorteBook}. Our implementation of the decoder
utilized the Blossom~V library due to Kolmogorov~\cite{kolmogorov2009blossom}.

Let us now show that Problem~2 is equivalent to the maximum weight matching
problem.  Indeed, for any pair of vertices $u,v\in \calV \setminus T$ let
$D_\phi(u,v)$ be the weighted distance between $u$ and $v$, that is, the
minimum weight of a path in $\calG$ that connects $u$ and $v$.  Likewise, let
$D_\phi(u,T)$ be the weighted distance between $u$ and the subset $T$.  One can
easily check that any minimum weight chain $E\subseteq \calE$ satisfying
$(\partial E)\backslash T=S$ consists of disjoint minimum weight paths that
connect pairs of vertices in $S$ and, possibly, 
minimum weight paths
that connect a vertex in $S$ with a vertex in $T$.  Hence Problem~2 is
equivalent to finding a (non-perfect) matching $M$ of vertices in the complete
graph $K_{|S|}$ with edge weights $D_\phi(u,v)$ that minimizes the objective function
\[
f(M)=\sum_{\mathrm{matched}} \; \; D_\phi(u,v) +  
\sum_{\mathrm{unmatched }} D_\phi(u,T).
\]
Here the first sum runs over all pairs of vertces $u,v\in S$
which are matched to each other in $M$, while the second sum runs over
 unmatched vertices $u\in S$. 
To minimize $f(M)$ define modified edge weights
\begin{equation}
\label{eta}
\eta(u,v)=D_\phi(u,T)+ D_\phi(v,T)-D_\phi(u,v).
\end{equation}
The objective function $f(M)$ can now be rewritten as
\[
f(M)=c-\sum_{(u,v)\in M} \eta(u,v), 
\]
where $c=\sum_{u\in S} D_\phi(u,T)$ is a constant that does not depend on the choice of $M$.
Thus Problem~2 is reduced to finding a maximum weight matching $M$
in the complete graph $K_{|S|}$ with edge weights $\eta(u,v)$.
A maximum weight
matching can be found in time  $O(|S|^3)$ using a slightly different version of Edmonds's blossom
algorithm~\cite{edmonds1965maximum}.
Our implementation of the decoder was based on 
the library LEMON~\cite{lemon} which realizes
a maximum weight matching algorithm due to Gabov~\cite{gabow1976efficient}.
It is worth pointing out that the optimal matching $M$
can be always chosen such that pairs of vertices with $\eta(u,v)\le 0$
are unmatched.  Such negative-weight edges can be safely removed from the
graph before calling the matching algorithm.
Although the reduction outlined above is elementary, to the best of our knowledge it has not
been used before in the context of error correction. 

\section{Correctability}
\label{sec:corr}

In this section we define correctable and uncorrectable error chains and prove
some technical lemmas needed for the analysis of the splitting method. 

Let $\calC_{min}(S)$ and $\calC_{min}(S,T)$ be the sets of 
minimum weight chains $E$ satisfying $\partial E=S$ and $(\partial E)\backslash T=S$
respectively, see Problems~1,2. Clearly, $\calC_{min}(S)=\calC_{min}(S,\emptyset)$.
Below we consider {\em minimum weight decoders} (MWD),
that is,  algorithms choosing a recovery chain
$R$ from the set $\calC_{min}(S,T)$  according to some specified rule.
Define the parity 
$\epsilon$ of an error chain $E\subseteq \calE$ as
\[
\epsilon(E)={|E\cap \Gamma| \pmod 2},
\]
where $\Gamma\subseteq \calE$ is the logical chain on the decoding graph.  We
will say that $E$ is odd (even) if $\epsilon(E)=1$ ($=0$).  Recall that error
correction is successful if the chosen recovery chain has the same syndrome and
the same parity as the actual error chain, see Section~\ref{sec:EC}.

Let $E\subseteq \calE$ be some fixed error chain.  Deciding whether a MWD
corrects $E$ is straightforward when all chains in $\calC_{min}(S,T)$ have the
same parity.  Indeed, in this case the correctability condition
$\epsilon(R)=\epsilon(E)$ is satisfied or not satisfied simultaneously for all
$R\in \calC_{min}(S,T)$.

In general, however, the set $\calC_{min}(S,T)$ may contain both even and odd
chains, see Fig.~\ref{fig:nasty} for a simple example.  In this case we will
say that $S$ is a {\em degenerate} syndrome.  Deciding whether an error chain
with a degenerate syndrome is corrected by a given MWD requires a detailed
knowledge of the rule used for choosing a recovery chain. 
As was pointed out by Stace and Barrett~\cite{StaceBarrett2010},
some decoding rules can reduce the logical error probability by breaking ties in favor
of chains having the largest entropy,  see Fig.~\ref{fig:nasty}.
 This demonstrates  that the logical error
probability $P_L$ is not well-defined unless the decoder's
behavior for all degenerate syndromes is specified. 

\begin{figure}[htb]
\includegraphics[scale=0.4]{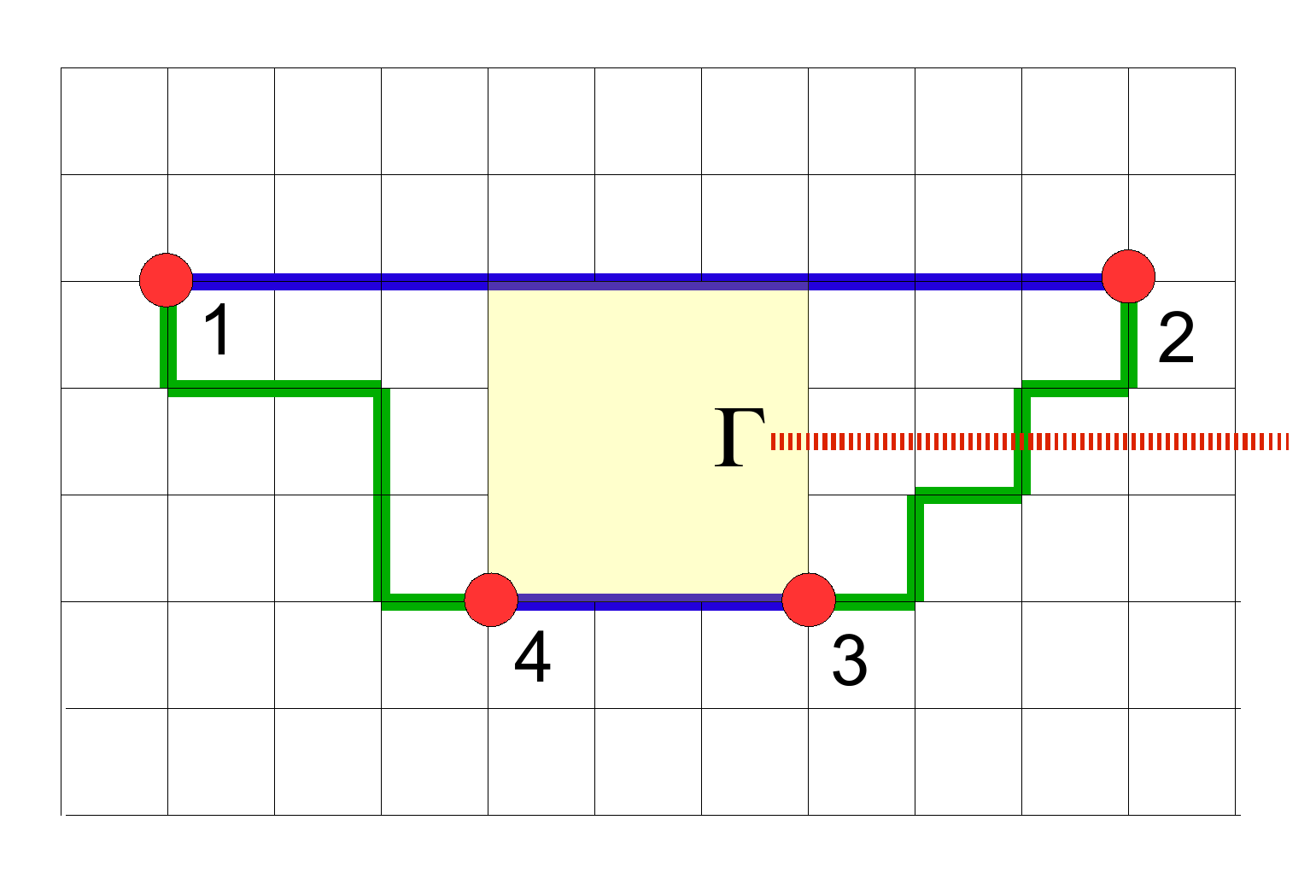}
\caption{\label{fig:nasty} 
{\em Example of a degenerate syndrome.} 
One can easily check that any error chain $E$ satisfying $\partial E=\{1,2,3,4\}$ has length at least $12$.
There is a unique even length-$12$ chain composed of paths $(1,2)$ and $(3,4)$.
There are ${6 \choose 3}^2=400$ odd length-$12$ chains composed of paths $(1,4)$ and $(2,3)$.
If the actual error chain has minimum length (which is true in the limit of small error rates), a decoding rule
that favors the pairing $(1,4),(2,3)$ is $400$ times more likely to correct the error.}
\end{figure}

To address this difficulty we opted to work with two different notions of
correctability --- decoder-specific and decoder-independent.  For the purposes
of numerical simulations, an error chain $E$ is called correctable iff a
recovery chain $R$ chosen by our implementation of the decoder satisfies
$\epsilon(R)=\epsilon(E)$. Otherwise, $E$ is called uncorrectable. Accordingly,
the logical error probability $P_L$ computed in our simulations applies only to
the one specific MWD implemented as described in Section~\ref{sec:EC}. From a
practical perspective, this is a natural approach, since an experimental
realization of error correction must be based upon some specific decoder.
Furthermore, we expect that $P_L$ does not vary too much for different MWDs
because degenerate syndromes are relatively rare.  The decoder-specific notion
of correctability was implicitly used in most of the previous work on the
subject. 

From the theoretical perspective, it is more natural to work with a
decoder-independent notion of correctability.
The definition given below is used in the rest of this section where we make
first steps towards a rigorous justification of the splitting method. 
\begin{dfn}
Let $E\subseteq \calE$ be an error chain with a syndrome
$S=(\partial E)\backslash T$.
We say that $E$ is correctable iff $\epsilon(E)=\epsilon(R)$
for all $R\in \calC_{min}(S,T)$.
Otherwise, $E$ is called uncorrectable.
\end{dfn} 
According to this definition, an error chain is correctable iff {\em any} MWD
corrects it. Note that this requirement is much stronger than the decoder-specific
correctability discussed above. 
Below we prove that the strong version of correctability can be tested efficiently
in certain special cases. The following lemmas are applicable to the decoding graphs
used in our simulations for noiseless syndrome readout, see Fig.~\ref{fig:decoding}.
Let us first assume that $T=\emptyset$.

\begin{lemma}
\label{lemma:1}
Suppose the decoding graph $\calG$ is planar with the maximum vertex
degree $O(1)$. Suppose that the logical chain $\Gamma$ is a path on the dual graph.
Then correctability can be tested in time $O(|\calV|^3)$.
\end{lemma}
\begin{proof}
Let $E$ be any error chain and $S=\partial E$ be its syndrome. 
Choose any recovery chain $R_0\in \calC_{min}(S)$.
It can be constructed in time $O(|\calV|^3)$, see Section~\ref{sec:EC}. 
If $R_0\oplus E$ has odd parity then $E$ is uncorrectable and we are done.
Below we assume that $R_0\oplus E$ has even parity.
Then $E$ is correctable iff all chains in $\calC_{min}(S)$ have the same parity.

A chain $C\subseteq \calE$ is called a cycle iff $\partial C=\emptyset$.
Let $\calC_{odd}$ be the set of all cycles with odd parity. 
Define a quantity 
\[
\delta=\min_{C\in \calC_{odd}} \; \phi(R_0\oplus C) - \phi(R_0).
\]
The minimality of $R_0$ implies that $\delta\ge 0$. 
We claim that $E$ is uncorrectable iff $\delta=0$. Indeed,
suppose $E$ is uncorrectable. Then there must exist a chain $R\in \calC_{min}(S)$
such that $R_0$ and $R$ have different parity. Therefore $C=R_0\oplus R$ is an odd cycle and 
$\phi(R_0\oplus C)=\phi(R)=\phi(R_0)$, that is, $\delta=0$. 
Conversely, assume that $\delta=0$. Then $\phi(R_0\oplus C)=\phi(R_0)$
for some odd cycle $C$ and thus $R=R_0\oplus C\in \calC_{min}(S)$
has parity different from $R_0$. Therefore $E$ is uncorrectable.

Thus, it suffices to show that the quantity $\delta$ can be computed in time $O(|\calV|^3)$.
Define new weights
\[
w(e)=\left\{ \begin{array}{rcl}
\phi(e) &\mbox{if} & e\notin R_0, \\
-\phi(e) & \mbox{if} & e\in R_0.\\
\end{array}
\right.
\]
Then $\delta$ coincides with the minimum $w$-weight
of an odd cycle, $\delta=\min_{C\in \calC_{odd}} w(C)$.
Let us show that $\delta$ can be expressed as the ground state energy
of the Ising model defined on a genus-$1$ graph (i.e. a graph embeddable into
a torus). This ground state energy can be computed in time $O(|\calV|^3)$ using algorithms
of Refs.~\cite{zecchina2001,valiant2002,cimasoni2007,bravyi2009contraction}.
Indeed, since $\calG$ is planar, it has a well-defined set of faces $\calV^*$ such that
 every edge $e\in \calE$ has exactly two adjacent faces $f,g\in \calV^*$.
 We will write $e=(f,g)$.
 For each face $f\in \calV^*$ introduce Ising spin $\sigma_f\in \{+1,-1\}$.
Given a spin configuration $\sigma=\{\sigma_f\}_{f\in \calV^*}$, define a chain $C(\sigma)\subseteq \calE$ such that
$C(\sigma)$ includes all edges $e=(f,g)$ with $\sigma_f\sigma_g=-1$.
The standard relationship between cycles in a planar graph and cuts in the dual graph
implies that $C(\sigma)$ is a cycle for any choice of $\sigma$ and that any cycle $C\subseteq \calE$
has form $C(\sigma)$ for some spin configuration $\sigma$.

By assumption, $\Gamma$ is a path on the dual graph connecting some pair of faces $f',f''\in \calV^*$.
Simple algebra shows that a cycle $C(\sigma)$ has odd parity iff
$\sigma_{f'}\sigma_{f''}=-1$. We conclude that $\delta=\min_{\sigma} H(\sigma)$,
where $H(\sigma)$ is the Ising-like Hamiltonian on the dual graph defined as
\[
H(\sigma)= J(1+\sigma_{f'} \sigma_{f''}) +\sum_{e=(f,g)\in \calE} w(e)(1-\sigma_f \sigma_g)/2
\] 
Here $J\gg 1$ is chosen large enough to guarantee that $\sigma_{f'}\sigma_{f''}=-1$ for any ground
state $\sigma$. To make the interaction $\sigma_{f'} \sigma_{f''}$ compatible with the graph
structure, the edge $(f',f'')$ must be added to the dual graph $\calG^*$.
Although the resulting graph is not planar, it is embeddable into a torus.
Indeed, one can first embed the primal graph $\calG$ into a sphere, create a pair of holes 
in the faces $f'$, $f''$, connect the two holes by a tube,
and finally draw the edge $(f',f'')$ on the tube's surface.
The algorithms of Refs.~\cite{zecchina2001,valiant2002,cimasoni2007,bravyi2009contraction}
enable exact computation of the partition function $Z(\beta)=\sum_\sigma \exp{[-\beta H(\sigma)]}$ in
time $O(|\calV|^3)$ for
any value of the inverse temperature $\beta$ 
assuming that the graph of spin-spin interactions has constant genus
(see, for instance, Theorem~1 of Ref.~\cite{bravyi2009contraction}).
The two cases $\delta=0$ and $\delta>0$ correspond to 
$\lim_{\beta\to \infty} Z(\beta)=1$ or $\lim_{\beta\to \infty} Z(\beta)=0$
which can be distinguished by computing $Z(\beta)$ 
for large enough $\beta$.
\end{proof}

Suppose now that $T$ is non-empty.
Recall that a subset of edges $C\subseteq \calE$ is called a {\em cut}
iff one can partition vertices of the graph into two disjoint sets,
$\calV=\calV_0\cup \calV_1$,  such that 
$C$ coincides with the set of edges connecting $\calV_0$ and $\calV_1$.
\begin{lemma}
Suppose the decoding graph $\calG$ is arbitrary and the logical chain 
$\Gamma$ is a cut of $\calG$.
 Then  correctability can be tested in time $O(|\calV|^3)$.
\end{lemma}
\begin{proof}
Let $E$ be any error chain and $S=(\partial E) \setminus T$ be its syndrome.
Choose any recovery chain $R_0\in \calC_{min}(S,T)$.  It can be constructed in
time $O(|\calV|^3)$, see Section~\ref{sec:EC}.  If $R_0\oplus E$ has odd parity
then $E$ is uncorrectable and we are done.  Below we assume that $R_0\oplus E$
has even parity.  Then $E$ is correctable iff all chains in $\calC_{min}(S,T)$
have the same parity.

Without loss of generality, no edge $e$ having both end-points in $T$ belongs
to $\Gamma$. Indeed,  such an edge $e$ does not participate in the constraint
$(\partial R)\backslash T=S$.  If $e$ has positive weight $\phi(e)$, no chain
in $\calC_{min}(S,T)$ contains $e$ and we can safely remove $e$ from $\Gamma$
without changing the parity of any $R \in \calC_{min}(S,T)$.  If $e$ has zero
weight, $\calC_{min}(S,T)$ contains both odd and even chains and the error is
uncorrectable. 

By assumption, $\calV=\calV_0\cup \calV_1$, $\calV_0\cap \calV_1=\emptyset$, and
$\Gamma$ is the set of edges connecting $\calV_0$ and $\calV_1$. 
Below we shall often use the obvious fact that for any chain $R\subseteq \calE$
the overlap $|R\cap \Gamma|$ is even (odd) iff $|(\partial R)\cap \calV_0|$ is even (odd).
Consider two cases. \\
{\em Case~1:} $T\cap \calV_0=\emptyset$. Then all chains $E$
satisfying $(\partial E)\backslash T=S$  are either even or odd
depending on whether $|S\cap \calV_0|$ is even or odd.
Hence $E$ is correctable. \\
{\em Case~2:} $T\cap \calV_0\ne \emptyset$. Define a new graph
$\tilde{\calG}=(\tilde{\calV},\tilde{\calE})$ obtained from $\calG$ by collapsing
all vertices of $T\cap \calV_0$ into a single vertex $t_0$. More precisely, we replace every edge 
$(u,v)$ with $u\in T\cap \calV_0$, $v\notin T$ by an edge $(t_0,v)$
without changing its weight. We can regard $\Gamma$ as a subset of edges
of $\tilde{\calG}$. 
 Define $S'=S$,  $S''=S\cup t_0$, and $\tilde{T}=T\cap \calV_1$.
Let $w_{min}'$ and $w_{min}''$
be minimum weights of chains $R',R''\in \tilde{\calE}$ 
 satisfying $(\partial R')\backslash \tilde{T}=S'$
and  $(\partial R'')\backslash \tilde{T}=S''$ respectively. Note that $w_{min}'$ and $w_{min}''$ can be
computed in time $O(|\calV|^3)$, see Section~\ref{sec:EC}.
Furthermore, any chain $R'$ as above has even (odd) parity iff
$|S\cap\calV_0|$ is even (odd). On the other hand, any chain $R''$ as above
has even (odd) parity iff $S\cap \calV_0$ is odd (even). Hence 
$w_{min}'=w_{min}''$ iff $\calC_{min}(S,T)$ contains two chains with a different parity.
We conclude that $E$ is uncorrectable iff $w_{min}'=w_{min}''$.
\end{proof}

Our implementation of the  splitting method involves 
a Metropolis-type subroutine for sampling uncorrectable
error chains from the chosen probability distribution. 
The following two lemmas are needed to prove that the
corresponding Markov process is ergodic.
Let us first assume that $T=\emptyset$.
\begin{lemma}
Suppose the decoding graph $\calG$  a square lattice with one smooth defect
and $\Gamma$ is a path on the dual lattice that connects the defect to the external boundary,
see Fig.~\ref{fig:decoding}.
Let $E,E'$ be any uncorrectable chains. 
Then one can transform $E$ to $E'$ by adding and removing
single edges such that all intermediate chains are uncorrectable.
\end{lemma}
\begin{proof}
Let $\Omega$ be the set of all edges lying on the boundary of the defect. 
Since $\Omega$ has trivial syndrome, $\partial \Omega=\emptyset$,
and odd parity, $\Omega$ is uncorrectable by Definition~1.
It suffices to prove the lemma for the special case $E'=\Omega$. 
We will need the following simple observation.
\begin{prop}
Let $S\subseteq \calV$ be a subset of vertices and 
$R\in \calC_{min}(S)$ be any chain. 
Choose any edge $e=(u,v)\in R$ and  define
$S'=S\oplus \{u,v\}$. Then 
$R\backslash e\in \calC_{min}(S')$. 
\end{prop}
\begin{proof}
Obviously, $\partial (R\backslash e)=\partial (R\oplus e)=\partial R \oplus \partial e=S'$.
Thus, $\partial(R\backslash e) = S'$ and $\phi(R\backslash e)=\phi(R)-\phi(e)$.
Suppose $R\backslash e$ is not in $\calC_{min}(S')$.
Then there exists $R'$ such that $\partial(R^\prime) = S^\prime$ and $\phi(R')<\phi(R)-\phi(e)$. 
Define $R''=R'\oplus e$. Then $\partial R'' = S$ and $\phi(R'')\le \phi(R')+\phi(e)<\phi(R)$
which contradicts the minimality of $R$. Thus $R\backslash e\in \calC_{min}(S')$. 
\end{proof}
Let $S=\partial E$ be the syndrome of $E$.
Since $E$ is uncorrectable, there must exist a recovery chain $R\in \calC_{min}(S)$
such that $E\oplus R=L_1\cup \ldots \cup L_m$ is a disjoint union of loops
and the number of odd loops among $L_1,\ldots,L_m$ is odd. 
Choose any edge $e\in R$ and define $E_1=E\oplus e$, $R_1=R\oplus e$. 
The proposition above implies that $R_1\in \calC_{min}(\partial E_1)$,
that is, $R_1$ is a minimum weight recovery chain for $E_1$.
Furthermore, since $E\oplus R=E_1\oplus R_1$, the new error chain $E_1$ is uncorrectable.
By repeating this argument one can construct a sequence of uncorrectable error chains
$E_0=E, E_1,\ldots, E_p$ such that 
$E_{i+1}$ is obtained from $E_i$ by adding (modulo two) any edge
from the recovery chain $R_i$ corresponding to $E_i$.
This process stops as soon as the recovery chain corresponding to $E_p$ 
is empty. At this step
$E_p=L_1\cup\ldots \cup L_m$ is a disjoint union of loops. 
We shall deal with the loops $L_\alpha$ one by one
such that each even loop is transformed into an empty chain (contracted)
while each odd loop is transformed into $\Omega$. 
These transformations can be implemented such that at every step 
a loop is modified by adding (modulo two) a boundary of some plaquette
which requires adding or removing at most three edges.  At each step
the syndrome 
 consists of at most two vertices and the recovery chain
consists either of a single edge or a pair of adjacent edges
lying on the boundary of some plaquette. 
Since the number of odd loops is odd, the final error chain coincides with $\Omega$.
\end{proof}
Suppose now that $T\ne \emptyset$.
\begin{lemma}
Suppose the decoding graph $\calG$ is a square lattice with two rough defects
and $\Gamma$ is a loop on the dual lattice encircling one of the defects,
see Fig.~\ref{fig:decoding}.
Let $E,E'$ be any uncorrectable chains. 
Then one can transform $E$ to $E'$ by adding and removing
single edges such that all intermediate chains are uncorrectable.
\end{lemma} 
\begin{proof}
As explained in the proof of Lemma~2, one can collapse all vertices $u\in T$ on the boundary
of each defect into a single vertex. Thus we can assume that $T=\{t',t''\}$.
Let $\Omega$ be a path connecting $t'$ and $t''$. 
Clearly, $\Omega$ is uncorrectable. 
It suffices to prove the lemma for the case $E'=\Omega$. 
Repeating the same steps as in the proof of Lemma~3
one can transform $E$ to a disjoint union of loops and paths connecting $t'$ and $t''$,
such that the number of paths is odd. Applying a sequence of plaquette transformations
as in the proof of Lemma~3 one can contract each loop and transform each path into $\Omega$.
\end{proof}

\section{Implementation of the splitting method}
\label{sec:implement}

Here we specialize the  splitting method described in Section~\ref{sec:split} to
the minimum weight decoding problem.
Let $\calG=(\calV,\calE)$ be the decoding graph 
corresponding to logical loop-like or path-like errors, see Section~\ref{sec:EC}.
For simplicity we shall first discuss the case when all edges
have the same error probability, $\pi(E)=p^{|E|}(1-p)^{n-|E|}$,
where $n=|\calE|$ is the total number of edges. 

We define an event as an arbitrary error chain  $E\subseteq \calE$. Accordingly,
$\Omega$ coincides with the set of all subsets of $\calE$. 
The subset of failure events  $\calF\subseteq \Omega$  consists of all uncorrectable chains
such that $\pi(\calF)$ is the probability of a logical loop-like or path-like error.
We shall first focus on the decoder-specific definition of correctability, see Section~\ref{sec:corr}.
Testing a membership $E\in \calF$ thus requires solving Problem~1 or~2.

We choose the family of distributions  $\pi_1,\ldots,\pi_t$ as
\[
\pi_j(E)=p_j^{|E|}(1-p_j)^{n-|E|}, \quad j=1,\ldots,t
\]
for a monotonic sequence of error rates $p_1,\ldots,p_t$ with $p_t=p$.
The following heuristic choice of the splitting sequence
was found to
provide a reasonable tradeoff between the statistical error and the number of splitting steps:
\begin{equation}
\label{steps1}
p_{j+1}=p_j 2^{\pm1/\sqrt{w_j}}, \quad w_j=\max{(d/2,p_jn)}.
\end{equation}
Here $d$ is the code distance and  the two signs correspond to upward and downward splitting.
To motivate this choice we note that
the quantity $w_j$ provides a rough estimate of  the average number 
of edges in a random chain $E\in \calF$ drawn from the distribution
$\pi_j(E|\calF)$. Since errors on different edges are independent and $p_j\ll 1$,
one should expect that the random variable $|E|$
is concentrated near its mean with the standard deviation $O(\sqrt{w_j})$.
Therefore
one can use a bound
\[
\left(\frac{p_{j+1}}{p_j}\right)^{-O(\sqrt{w_j})} \le \frac{\pi_{j+1}(E|\calF)}{\pi_j(E|\calF)} \le
\left(\frac{p_{j+1}}{p_j}\right)^{O(\sqrt{w_j})}
\]
for all `typical' error chains $E$. Here, for concreteness, we  consider upward splitting, that is, $p_j<p_{j+1}$. Then Eq.~(\ref{steps1}) implies
\begin{equation}
\label{steps2}
c^{-1}\pi_j(E|\calF)\le \pi_{j+1}(E|\calF) \le c\pi_j(E|\calF)
\end{equation}
for some constant $c=O(1)$.
Thus $\|\pi_j(\cdot|\calF)-\pi_{j+1}(\cdot|\calF)\|_1\le 1-1/c$
and  the ratio $\pi_{j+1}(\calF)/\pi_j(\calF)$ is estimated with 
an error $\sigma_j\le \sqrt{2cN^{-1}}=O(N^{-1/2})$, see Eq.~(\ref{SM9}).

In order to sample a chain $E\in \calF$ from the conditional distribution
$\pi(E|\calF)\equiv \pi(E)/\pi(\calF)$ we used a  Metropolis-type
subroutine. A single Metropolis step  
takes as input a chain $E\in \calF$ and outputs a new chain $E'\in \calF$
which differs from $E$ on at most one edge as described below.
\begin{enumerate}
\item Select an edge $e\in \calE$ at random from the uniform distribution.
Set $E'=E\oplus e$. 
\item Compute $q=\min{[1,\pi(E')/\pi(E)]}$ and generate a random
bit $b=0,1$ such that $\mathrm{Pr}(b=1)=q$.
\item If $b=0$ then stop and output $E$.
\item If $b=1$ and $E'\in \calF$  then output $E'$. 
Otherwise output $E$.
\end{enumerate}
For any chains $E,E'\in \calF$ let $P(E,E')$ be the probability that the Metropolis step
outputs $E'$ if called on the input $E$. 
One can easily check that $P$ obeys a detailed balance condition
\begin{equation}
\label{DB}
\pi(E)P(E,E')=\pi(E')P(E',E) \quad \mbox{for all $E,E'\in \calF$}.
\end{equation}
Thus the Metropolis step defines
a reversible Markov process  $\calM$ such that states of $\calM$
are uncorrectable error chains,  $P(E,E')$ is the transition probability from
$E$ to $E'$, and $\pi(E|\calF)$ is a steady distribution of $\calM$. 
A similar Markov process $\calM_j$ is constructed for each distribution
$\pi_j$ in the splitting sequence. 

The full Metropolis subroutines involves $M\gg 1$ Metropolis steps
starting from some fixed initial uncorrectable chain $E_0$.
The latter was chosen as a loop encircling a defect (for loop-like errors)
or a path connecting the two defects (for path-like errors).
The sequence of uncorrectable error chains $E_0,E_1,\ldots,E_M$ generated by the Metropolis subroutine 
was used to estimate the expectations values in Eq.~(\ref{SM5}).
In order for Eq.~(\ref{SM5})  to hold, the number of Metropolis steps
must satisfy $M\gg N\tau_j$, where $\tau_j$ is the mixing time of the Markov
process $\calM_j$. 
Since in practice the mixing time is unknown,  the number of steps
$M$ was chosen by checking two heuristic conditions: (i)  statistical fluctuations of the 
righthand side of Eq.~(\ref{SM5}) are smaller than the desired precision, (ii) doubling $M$
does not change the  righthand side of Eq.~(\ref{SM5}) by more than the desired precision. 
Our goal was to compute the logical error probability $\pi(\calF)$ with 
a relative error about $50\%$. Note that $\pi(\calF)$ changes by almost $20$
orders of magnitude in our simulations. A relative error $50\%$ is thus good enough
for all practical purposes. 
 Accordingly, we aimed at estimating the expectation
values in Eq.~(\ref{SM5}) with a relative error around $0.5/t$, where $t$ is the number 
of splitting steps. 
The required number of Metropolis flips (non-trivial steps) 
varied in the range $10^5$ to $10^7$ depending on the geometry, lattice dimensions,
and the error rate.

Note that some Metropolis steps may require
testing a membership $E'\in \calF$, which in turn requires
 solution of Problem~1 or Problem~2, see  Section~\ref{sec:EC}.
As the full Metropolis subroutine may involve millions of steps, a natural question
is whether the solution of Problems~1,2 
obtained at some Metropolis step $j$ can be `recycled' and used at the next step $j+1$.
Suppose the corresponding error chains differ on some edge $e$,
$E_{j+1}=E_j\oplus e$. To perform the standard reduction from Problems~1,2 to the minimum (maximum)
weight matching problem one has to construct a family of   minimum weight paths on the decoding
graph connecting any pair of syndrome vertices, see Section~\ref{sec:EC}.
Our implementation of the Metropolis subroutine recycles 
minimum weight paths found for the syndrome $S_j=\partial E_j\backslash T$
and uses them 
to construct minimum weight paths for 
the syndrome $S_{j+1}=\partial E_{j+1}\backslash T$.
This yields a significant speedup since the syndromes 
 $S_j$ and $S_{j+1}$ differ on at most two vertices.

In the case of loop-like logical errors  simulations were performed only for the 
single-defect geometry. This is a natural simplification since the decay rate $\alpha$ 
for loop-like errors does not depend on the
number of defects, as long as separation between defects is sufficiently large. 
In the case of path-like logical errors simulations were performed for the double-defect
geometry. The corresponding decoding graphs are shown on Fig.~\ref{fig:decoding}.

The splitting method can be similarly applied  to the minimum weight decoding
problem with a decoder-independent notion of correctability, see Section~\ref{sec:EC}.
In particular, for the decoding graphs shown on Fig.~\ref{fig:decoding}
the Metropolis step can be implemented in time $O(|\calV|^3)$ and the corresponding Markov process
$\calM$ is ergodic, see Lemmas~1-4.
Therefore,  in these special cases the conditional distribution $\pi(E|\calF)$ is the unique
steady state of $\calM$ and our results rigorously
prove correctness of the splitting method in the limit $M\to \infty$.

\section{Noisy syndrome readout}
\label{sec:noisy}

Here we describe the construction of the decoding graph and the implementation
of the splitting method in the case when the syndrome readout circuit itself
may introduce errors. The material of this section is mostly based on
Refs.~\cite{Dennis01,Fowler08,fowler2012topological}.

We begin by defining the noise model and the syndrome readout circuit.
Our set of elementary operations includes
CNOT gates, single-qubit measurements in the $X$- or $Z$-basis, and
preparation of single-qubit ancillary states $|0\rangle$ and $|+\rangle$.
The syndrome readout circuit consists of a sequence of rounds, 
where at each round any qubit can participate in one elementary operation
or remain idle. 
Each elementary operation can fail
with a probability $p$ that we call an {\em error rate}. More precisely,
our error model, borrowed from~\cite{Fowler08},  is defined as follows.
\begin{itemize}
\item A noisy $X$ or $Z$ measurement is the ideal measurement in which the outcome is flipped
with probability $p$.
\item A noisy $|0\rangle$ or $|+\rangle$ ancilla preparation returns the correct state with probability $1-p$ and the
 orthogonal state $|1\rangle$ or $|-\rangle$  with probability $p$.
\item A noisy CNOT gate is the ideal CNOT gate followed by
one of $16$ two-qubit Pauli operators $P$. We apply $P=I$ with probability $1-p$
and each individual $P\ne I$ with probability $p/15$.
\item If a qubit remains idle during some round, it is acted upon by $X$, $Y$ or $Z$
error with probability $p/3$ each (``memory error''). 
\end{itemize}

\begin{figure}[htb]
\centerline{\includegraphics[height=7cm]{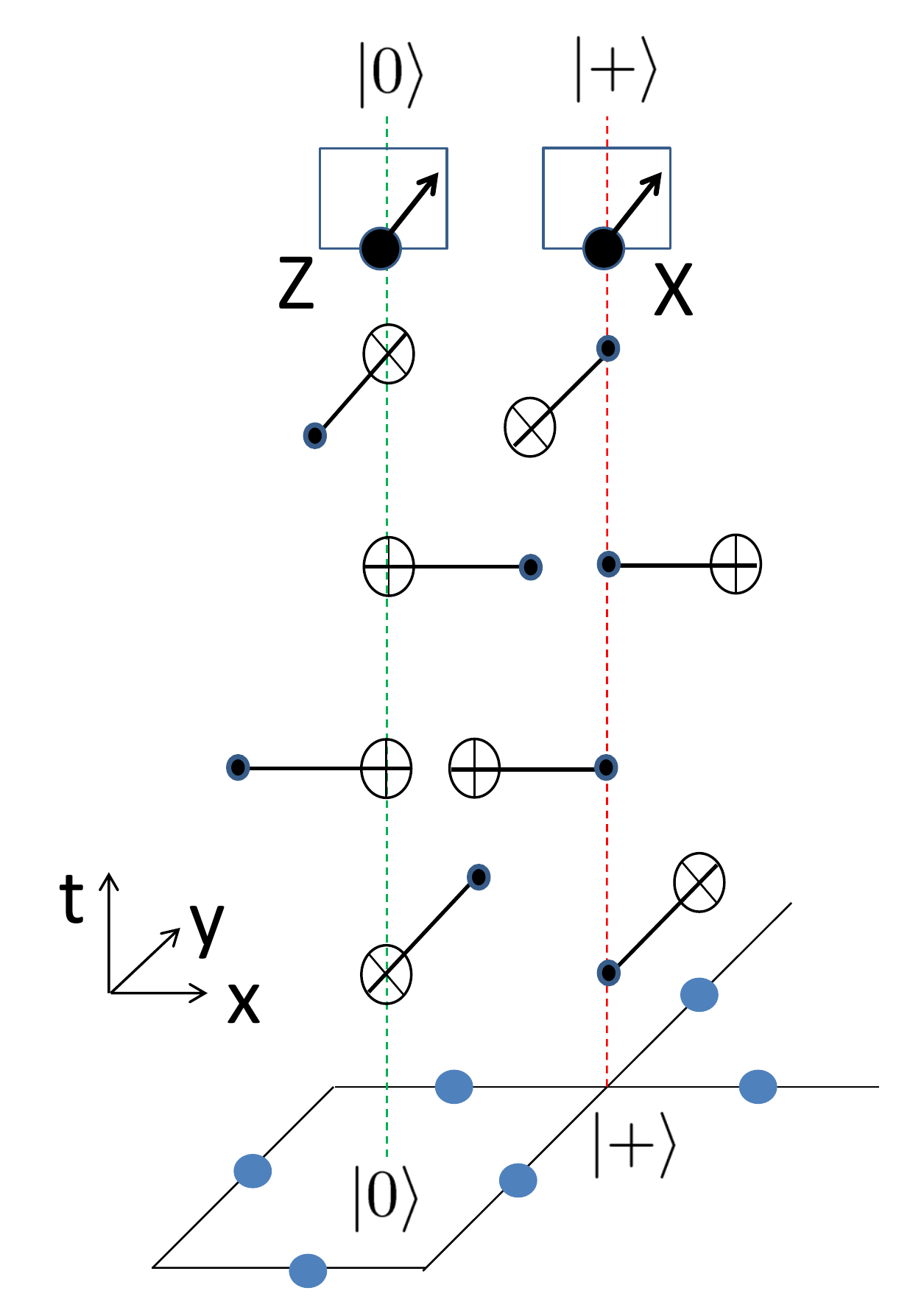}}
\caption{(Color online) Syndrome readout circuits for plaquette and site stabilizers.\label{fig:readout}}
\end{figure}

Following Refs~\cite{Dennis01,Fowler08}, we measure eigenvalues of site and
plaquette stabilizers using the quantum circuit shown in
Fig.~\ref{fig:readout}.  Measuring a single stabilizer requires one ancillary
qubit and six rounds. In the case of truncated stabilizers located near the
boundary of defects, some CNOT gates in the circuit of Fig.~\ref{fig:readout}
are skipped and the corresponding ancillary qubits remain idle. The ancillary qubits are
located at the centers of plaquettes and at sites of the physical lattice.  The
syndrome readout is repeated periodically in time until enough syndrome data is
collected to enable reliable error correction, see below.  We assume that the
rounds are scheduled such that a new set of syndrome data from every stabilizer
arrives at each integer time step $t$ (accordingly, each round takes $1/6$
units of time). For simplicity, we assume that 
defects are not added, changed, or annihilated during the collection of 
syndrome data; that is, we are only examining storage of information in time.

The error correction protocol is tailored to the chosen noise model and the
syndrome readout circuit~\cite{Fowler08}. As before, the key ingredient in the
protocol is a decoding graph $\calG=(\calV,\calE)$ and we again need two
independent decoding graphs for dealing with phase-flip and bit-flip errors
separately. For concreteness, below we focus on bit-flip errors.  The vertices
of $\calG$ can be partitioned into two disjoint subsets, $\calV = V \cup T$.
Each vertex in $V \subseteq \calV $ is a pair $u=(p,t)$ that represents the
space-time location of a syndrome bit measured at time step $t$ at a plaquette
$p$.  Let $S\subseteq \calV$ be the set of all vertices $(p,t)$ such that the
syndrome bits measured at the plaquette $p$ at time steps $t$ and $t+1$ are
different.  We shall refer to $S$ as a {\em relative syndrome}.  Clearly,
$S=\emptyset$ in the absence of errors. Suppose now that the syndrome readout
circuit contains a single {\em error event}, that is, any non-identity Pauli
operator applied as an error in any single elementary operation or idle step.
It can be shown that a relative syndrome caused by any single error event in
the circuit consists of at most two vertices, see~\cite{Fowler08}.  We connect
a pair of vertices $u,v\in \calV$ by an edge iff the relative syndrome
$S=\{u,v\}$ can be created by a single error event. Some error events located
near the boundary of a defect create a relative syndrome at a single vertex
$u$. Such an error event is represented on the decoding graph by a ``hanging
edge" attached to $u$.  Each vertex $u \in V$ is attached to at most one
hanging edge; the other endpoint of this hanging edge has degree one and is an
element of $T \subseteq \calV$.  The set $T$ solely consists of the hanging
edge endpoints that are not in $V$.  

The decoding graph corresponding to the
surface code lattice of Fig.~\ref{fig:flat} for bit-flip errors is shown on
Fig.~\ref{fig:3D}.  To avoid clutter, we represent a hanging edge attached to
some vertex $u \in V$ by a solid circle centered at $u$.  A typical vertex of the
decoding graph has $12$ incident edges, see Fig.~\ref{fig:3Dzoom} for detail.
Generally, memory errors are represented by edges oriented along the $x$ or $y$
axes, while measurement and initialization errors lead to edges oriented along
the $t$ axis.  Any type of edge can be observed due to the two-qubit errors
(which occur after each of the rounds of CNOT gates), and many of the diagonal
edges can only be observed in this way.  Two edges of the same orientation may
be created by different errors depending on their proximity to defects.  The
decoding graph describing phase-flip errors is constructed in a similar fashion
(hanging edges on the phase-flip decoding graph are not needed for the defects examined here).

By construction, every edge $e$ of the decoding graph represents some set of
error events $\Omega_e$ in the syndrome readout circuit (those that create
relative syndromes at the endpoints of $e$).  The sets $\Omega_e$ corresponding
to different edges are disjoint.  Define a {\em prior} $p(e)$ as the
probability of observing an {\em odd} number of error events in $\Omega_e$ upon
execution of the circuit (since errors add up modulo two, only the parity of
the number of errors matters).  As argued in~\cite{Fowler08}, knowledge of the
priors significantly improves decoding success probability.  Edges of the
decoding graph with large prior probabilities can be regarded as more noisy and
should be preferred over less-noisy edges when choosing a recovery chain.  The
time-like edges and the diagonal edges are the most and the least noisy
respectively. We estimated the priors $p(e)$ by summing up the
probabilities of all error events in the set $\Omega_e$.  The priors are
represented by a color scale on Fig.~\ref{fig:3D}.

%
Any combination of error events in the circuit can be represented by an error
chain $E\subseteq \calE$ in the decoding graph such that $e\in E$ iff the set
$\Omega_e$ contains an odd number of the error events that occurred when the
circuit was run.  Given an edge $e\in \calE$, let $\Pi(e)\subseteq \Sigma_1$ be the
`projection' of $e$ onto the 2D surface code lattice.  More precisely, if
$e=(u,u')$ for some vertices $u=(p,t)$ and $u'=(p',t')$ then $\Pi(e)$ consists
of the edges making up a minimum weight path between $p$ and $p^\prime$
for $p\ne p'$ and $\Pi(e)=\emptyset$ otherwise.  Note that when $p \ne p^\prime$, 
$\Pi(e)$ contains one edge except in some of the cases where $e$ is a diagonal edge.
Given any chain
$E=\{e_1,e_2,\ldots,e_m\}\subseteq \calE$ on the decoding graph, the
corresponding accumulated 2D error chain on the surface code qubits is 
\[
\Pi(E)=\Pi(e_1)\oplus \Pi(e_2) \oplus \ldots \oplus \Pi(e_m)\subseteq \Sigma_1.
\]
Our goal is to use the syndrome information to correct the error chain
$\Pi(E)$.  By construction, $(\partial E)\backslash T =S$, where $S\subseteq
\calV$ is the relative syndrome and $T\subseteq \calV$ is defined above.  Thus
finding the most likely error chain consistent with a given relative syndrome
$S$ is equivalent to solving Problem~1 or~2, see Section~\ref{sec:EC}.  Let
$R\in \calC_{min}(S,T)$ be a minimum weight recovery chain constructed by the
decoder.  The recovery operator corresponding to $R$ is determined by the 2D
projection $\Pi(R)$.  An error chain $E$ is called correctable iff
$\Pi(R)\oplus \Pi(E)$ has even overlap with the relevant logical chain
$\Gamma$, see Section~\ref{sec:corr}.  Equivalently, $E$ is correctable iff
$R\oplus E$ has even overlap with a 3D logical chain $\hat{\Gamma}\subseteq
\calE$ that includes all edges $e\in \calE$ such that $\Pi(e)\cap \Gamma$ is not empty. In
the case of bit-flip errors, $\hat{\Gamma}$ is the set of all hanging edges
located on the boundary of the left defect tube, see Fig.~\ref{fig:3D}. For
phase-flip errors, $\hat{\Gamma}$ can be visualized as a `membrane' connecting
the defect tube to the external spatial boundary of the lattice (not shown).
Let $\Omega$ be the set of all error chains and $\calF\subseteq \Omega$ be the set of uncorrectable chains.

Apart from the different definition and interpretation of the decoding graph,
the implementation of the splitting method is exactly the same as described in
Section~\ref{sec:implement}.
Our simulations were performed for a phenomenological noise model
where errors on different edges of the decoding graph occur independently
with probabilities $p(e)$. 
In order to evaluate the quantity $P_L(p)$ for a
given probability $p$, we used a family of distributions $\pi_1,\ldots,\pi_t$
defined as
\[
\pi_j(E)=\prod_{e\in E} p_j(e) \prod_{e\in \calE\backslash E} (1-p_j(e)),
\]
where $p_j(e)$ are the priors computed for a monotone decreasing sequence of error rates
$p_1,\ldots,p_t$ such that $p_t=p$. The sequence $p_1,\ldots,p_t$ is defined by the heuristic
rule Eq.~(\ref{steps1}), where $w_j=\sum_{e\in \calE} p_j(e)$.
Simulations were performed  only
for defects with linear size $r=2,3,4$, 
partly due to the growing running time of the Metropolis subroutine and partly due to
computer memory limitations
(for the double defect geometry with $r=4$ the lookup table of minimum weight paths on the decoding graph
takes about 4GB of RAM).  For path-like logical errors  the decoding graph shown
 in Fig.~\ref{fig:3D} is a 3D lattice with a pair
 of vertical defect tubes. The decoding graph corresponding to loop-like logical errors
 is similar to Fig.~\ref{fig:3D}, but there is only one vertical defect tube.
By combining the splitting method and the Monte Carlo data 
we were able to compute parameters of the fitting formula Eqs.~(\ref{Fit},\ref{xy})
which we expect to be valid for larger code distances. 
The decay rate $\alpha(p)$ in the exponential scaling
$P_L\sim \exp{[-\alpha(p)r]}$ is shown in Fig.~\ref{fig:alpha}.

\section{Conclusion and open problems}
\label{sec:concl}

We proposed a new algorithm for estimating the logical error probability 
of the surface code in the regime of large code distances
and moderately small error rates. Numerical results are presented for two commonly studied error
models corresponding to noiseless and noisy syndrome extraction. 
Our results demonstrate that the asymptotic formulas for the logical error probability $P_L(p)$
valid in the limit $p\to 0$  tend to underestimate $P_L(p)$ for finite error rates. 
A more accurate fitting formula for $P_L(p)$ is proposed. 

Our work certainly leaves many important questions unanswered. 
First, one may ask whether our simulation techniques can be extended to non-trivial logical
gates, such as the CNOT gate, or more complicated logical circuits such as the topological state
distillation~\cite{FowlerBridge}. Each of these circuits can be visualized as a network of defect tubes
embedded into a 3D space-time~\cite{FowlerBridge}. We anticipate that the logical error probability $P_L$
associated with a large network of tubes can be estimated by decomposing the network into small tiles that 
consist of single isolated tube segments or parallel pairs of such segments. The techniques presented in this paper are applicable to 
each individual tile. Therefore one can get a rough estimate of $P_L$ by  summing up  logical error probabilities associated with each tile.

From the theoretical perspective, it is desirable to derive rigorous bounds on the
running time and the approximation error of the algorithm. This, in turn, requires
upper bounds on the mixing time of the Metropolis subroutine described in Section~\ref{sec:implement}.
We conjecture that the mixing time scales as $p^{-\Omega(d)}$ for a general distance-$d$
surface code with multiple defects in the limit $p\to 0$. The intuition behind this conjecture is that 
minimum-weight uncorrectable error chains that are localized on the boundary of different
defects cannot be connected by a sequence of local Metropolis steps without passing
 through intermediate high-weight uncorrectable  error chains.
However, if the lattice contains a single defect (in the case of loop-like errors)
or a pair of defects (in the case of path-like errors), see Fig.~\ref{fig:decoding},
it is plausible that the mixing time is a sub-exponential function of $d$. 

One can also explore possible generalizations of our algorithm
to different noise models, such as the true circuit-based noise model, see Ref.~\cite{Fowler08},
and different stabilizer codes. Finally, we expect that our fitting formula
for the logical error probability can be refined by  taking
into account the pre-exponential factor depending on $d$
as was proposed in Ref.~\cite{RHG07}.

\begin{figure*}[h]
\includegraphics[scale=0.65]{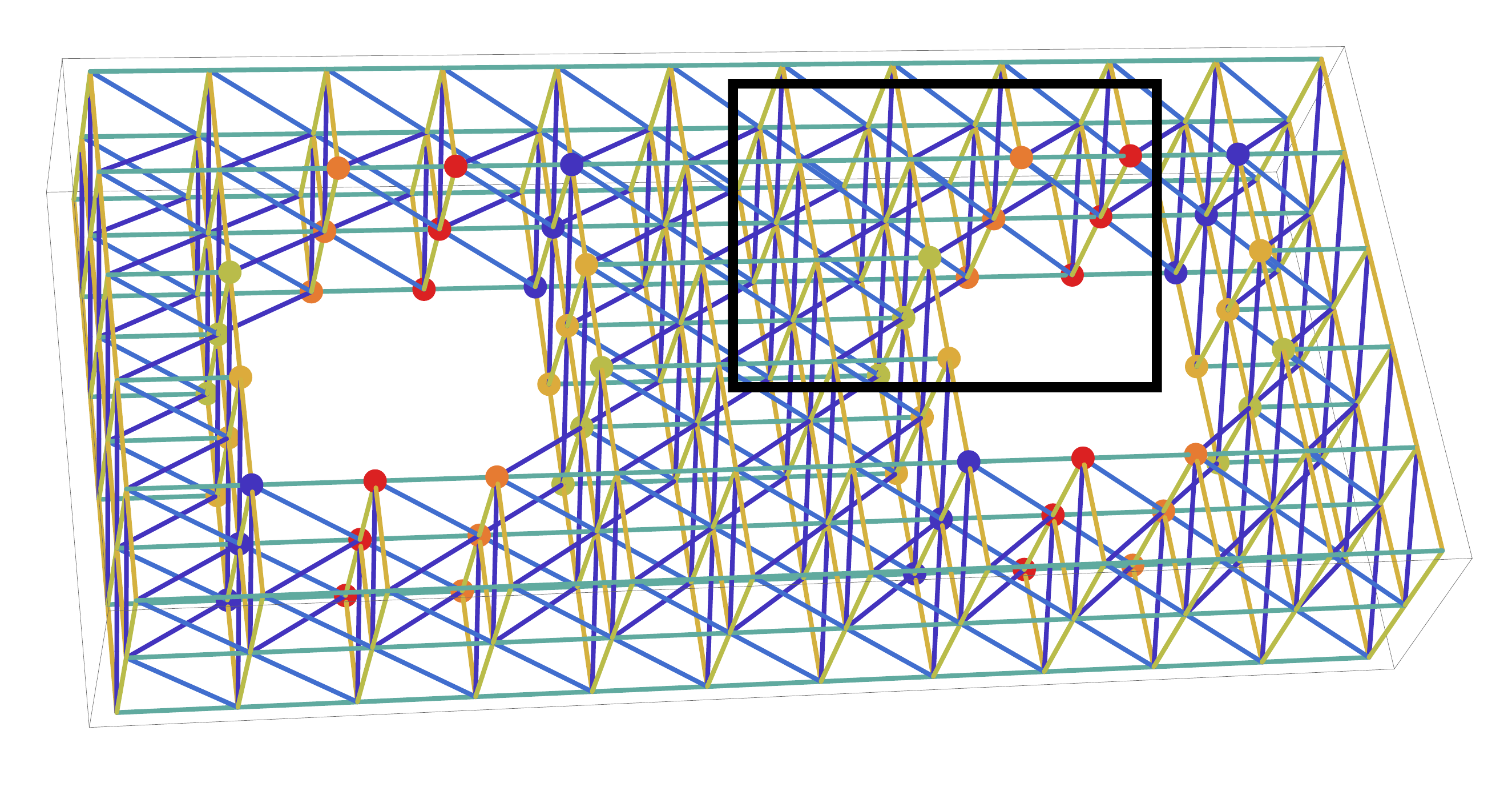}
\caption{\label{fig:3D}
(Color online) {\em Noisy syndrome readout.} 
The decoding graph corresponding to
the surface code lattice shown in Fig.~\ref{fig:flat}.
Each vertex of the graph represents a space-time location of a syndrome measurement.
Any elementary error in the syndrome readout circuit (memory, preparation, measurement, or CNOT error)
is associated with some edge of the graph. The overall probability of elementary errors
associated with a given edge $e$ determines the effective error rate of $e$. Red (blue) color stands for the
largest (smallest) effective error rate. Some edges located near the boundary of
the defects have only one end-point. To avoid clutter, such edges are represented by solid circles. 
The corresponding surface code has two smooth defects 
of linear size $r=2$, separation $s=5$, and buffer length $b=3$.
The number of syndrome readout rounds is $t=3$.
The actual simulations were performed for $s=b=t=4r$ for $r=2,3,4$.
}
\end{figure*}

\begin{figure}[h]
\includegraphics[scale=0.3]{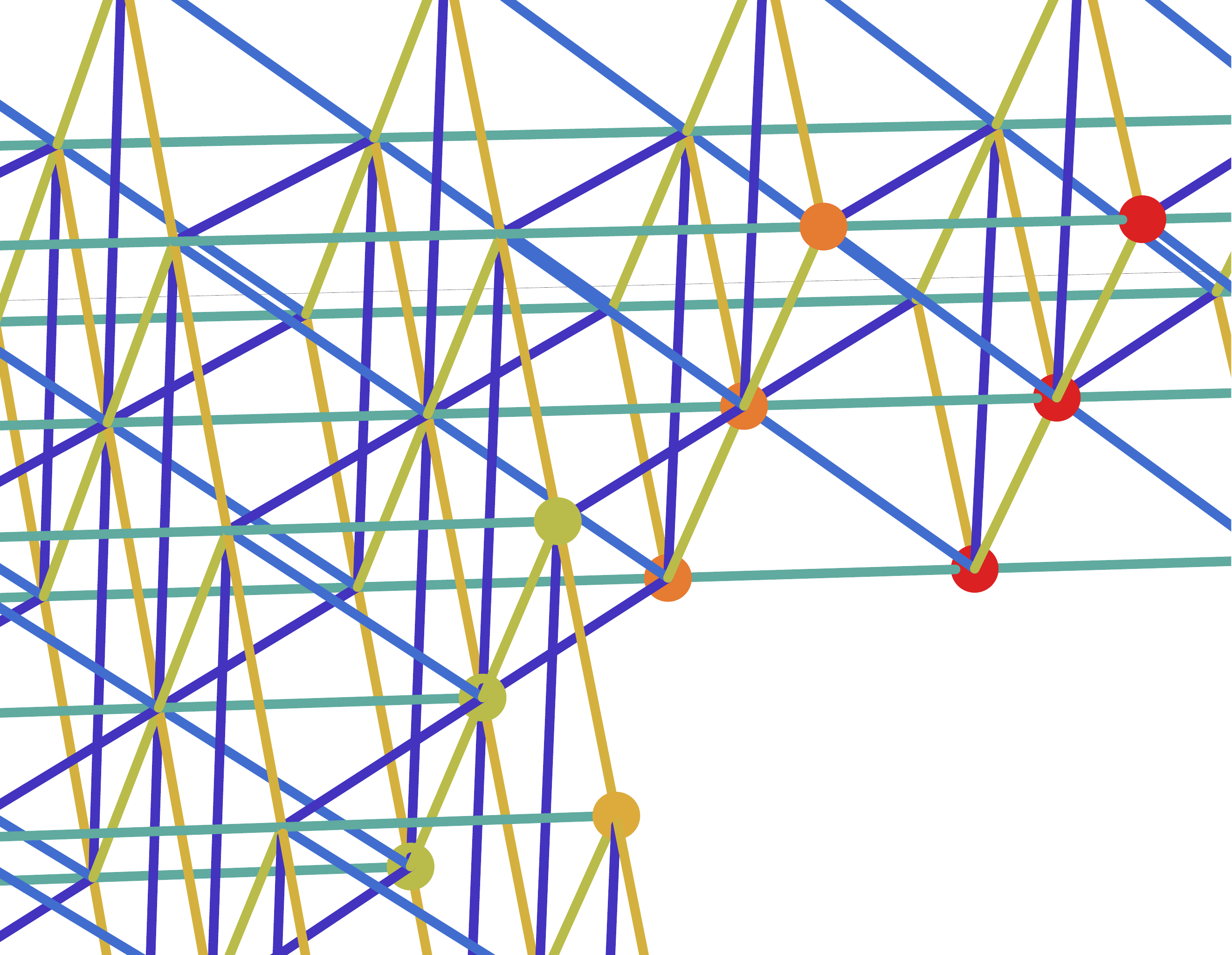}
\caption{\label{fig:3Dzoom}
A fragment of the decoding graph indicated by the black rectangle on Fig.~\ref{fig:3D}.
}
\end{figure}

\begin{center}
{\bf Acknowledgments}
\end{center}
We would like to thank Charles Bennett, Martin Suchara, and
Rajan Vadakkedathu for helpful discussions.
This work was supported in part by 
IARPA QCS program under contract number D11PC20167
and  by the DARPA QuEST program under contract number HR0011-09-C-0047.
Computational resources for this work were provided by IBM Blue Gene Watson supercomputer center.


\end{document}